\documentclass[reqno]{amsart}

\usepackage[foot]{amsaddr}

\makeatletter
\renewcommand{\email}[2][]{%
  \ifx\emails\@empty\relax\else{\g@addto@macro\emails{,\space}}\fi%
  \@ifnotempty{#1}{\g@addto@macro\emails{\textrm{(#1)}\space}}%
  \g@addto@macro\emails{#2}%
}
\makeatother

\usepackage[sort&compress,numbers]{natbib}
\usepackage[utf8x]{inputenc}
\usepackage{amsmath,amsfonts,amssymb,mathrsfs}
\usepackage[]{hyperref}
\hypersetup{
  colorlinks   = true,
  urlcolor     = blue,
  linkcolor    = blue,
  citecolor   = red
}

\newtheorem{theorem}{Theorem}
\theoremstyle{plain}

\newtheorem{corollary}{Corollary}
\newtheorem{definition}{Definition}

\newtheorem{lemma}{Lemma}

\newtheorem{remark}{Remark}
\numberwithin{equation}{section}

\newcommand{\naturals}{\mathbb{N}}
\newcommand{\reals}{\mathbb{R}}
\newcommand{\complexes}{\mathbb{C}}

\renewcommand{\Re}[1]{\operatorname{Re}{#1}}

\newcommand{\spec}[1]{\operatorname{spec}{#1}}
\newcommand{\proj}[1]{\mathsf{P}_{#1}}
\newcommand{\tr}[1]{\operatorname{tr}{#1}}

\newcommand{\id}[1]{\mathrm{id}_{#1}}
\newcommand{\comm}[2]{[#1,#2]}
\newcommand{\acomm}[2]{\{#1,#2\}}
\newcommand{\dual}{*}
\newcommand{\hadj}{*}

\newcommand{\hsiprod}[2]{\langle #1, #2 \rangle_{\mathrm{HS}}}
\renewcommand{\vec}[1]{\mathbf{#1}}

\newcommand{\supnorm}[1]{\left\| #1 \right\|_{\infty}}

\newcommand{\matr}[1]{\mathbb{M}_{#1}}
\newcommand{\matrd}{\matr{d}}

\newcommand{\choi}[1]{\mathcal{C}(#1)}

\newcommand{\cpe}[1]{\operatorname{\mathsf{CP}}(#1)}

\newcommand{\cocpe}[1]{\operatorname{\mathsf{coCP}}(#1)}

\newcommand{\ppte}[1]{\operatorname{\mathsf{PPT}}(#1)}
\newcommand{\ebe}[1]{\operatorname{\mathsf{EB}}(#1)}
\newcommand{\transpose}{\mathsf{T}}
\newcommand{\ptranspose}[1]{\transpose_{#1}}
\newcommand{\ptrans}{\ptranspose{2}}
\newcommand{\atime}[1]{\tau_{\mathrm{#1}}}

\begin{document}
\title[Eventually entanglement breaking...]{Eventually entanglement breaking divisible quantum dynamics}

\author{Krzysztof Szczygielski}
\address[K.~Szczygielski]
{Institute of Theoretical Physics and Astrophysics, Faculty of Mathematics, Physics and Informatics, Wita Stwosza 57, 80-308 Gda\'{n}sk, Poland and Institute of Physics, Faculty of Physics, Astronomy and Informatics, Nicolaus Copernicus University, Grudzi\c{a}dzka 5/7, 87–100 Toruń, Poland}
\email[K.~Szczygielski, corresponding author]{krzysztof.szczygielski@ug.edu.pl}

\author{Dariusz Chru\'sci\'nski}
\address[D.~Chru\'sci\'nski]
{Institute of Physics, Faculty of Physics, Astronomy and Informatics, Nicolaus Copernicus University, Grudzi\c{a}dzka 5/7, 87–100 Toruń, Poland}
\email[D.~Chru\'sci\'nski]{darch@fizyka.umk.pl}

\begin{abstract}

It is shown that a large class of quantum dynamical maps on complex matrix algebras governed by time-local Master Equations tend to become entanglement breaking in the course of time. Such situation seems to be generic for quantum evolution and in particular, completely positive dynamical semigroups with a unique faithful stationary state enjoy this property. Inspired by this observation, we propose a new concept of eventually entanglement breaking divisible (eEB-divisible) dynamics. A dynamical map is eEB-divisible if any propagator becomes entanglement breaking in finite time. It turns out that eEB-divisibility is quite general and holds for a large class of quantum evolutions.

\end{abstract}

\maketitle

\section{Introduction}

In this article we explore certain asymptotic aspects of quantum evolution families, or quantum dynamical maps, on algebra $\matrd$ of complex square matrices of size $d\geqslant 2$. We make a following observation: \emph{a large class of evolution families, governed by a time-local Master Equation, exhibits a tendency of becoming entanglement breaking maps, either in finite time or asymptotically.} Such phenomenon appears to emerge in quite a natural, generic manner, both in Markovian and (weakly) non-Markovian scenarios. We show it rigorously in few simplified cases, which include dynamics governed by commuting time-dependent generators, as well as quantum dynamical semigroups. In particular, a general result for the latter is that the semigroup becomes entanglement breaking in finite time if its generator admits a one-dimensional kernel spanned by a strictly \emph{positive definite} stationary state (i.e.~of full rank). Such sufficient condition may be then generalized beyond semigroup case. Furthermore, the tendency of becoming entanglement breaking is also observed for \emph{propagators} of evolution families in a number of cases. This in turn justifies a new notion of divisibility which we propose, the \emph{eEB-divisibility}, where a defining property is such that the propagator itself becomes entanglement breaking in finite time. The eEB-divisibility property constitutes for another sufficient, yet not necessary, condition for a family to become entanglement breaking.

The article is structured as follows. We start with a short introduction to theory of positive maps in Section \ref{sec:PositiveMaps}. In Section \ref{sec:Asymptotics} we define the notion of \emph{eventually entanglement breaking} families and give its simple characterization in terms of spectral properties of generators. The succeeding Section \ref{sec:EBdivisibility} is then devoted to eEB-divisibility and its connection to asymptotic behavior of evolution families. This is then further explored in CP-divisible case in Section \ref{sec:CPdivisible}. Next, in Section \ref{sec:PPT2} we make a note on a connection to the famous PPT\textsuperscript{2}-conjecture, here in asymptotic sense. Finally, in Sections \ref{sec:CasesQubit} and \ref{sec:CasesBeyondQubit} we focus on some most distinguished classes of quantum dynamical maps, including CP-divisible evolution and quantum dynamical semigroups.

\section{Positive maps on matrix algebras}
\label{sec:PositiveMaps}

We start with a brief recollection of basic facts about positive maps between matrix algebras. For sake of clarity and brevity we will highlight only necessary concepts and we refer the Reader to vast literature (see e.g.~\cite{Stoermer1963,Paulsen2003}) for comprehensive study of the subject. Throughout the article, $\matr{n}$ will be a C*-algebra of $n\times n$ matrices over $\complexes$ endowed with spectral matrix norm and Hermitian conjugation as involution; $\matr{n}^+$ will denote a closed convex cone of positive semi-definite matrices (we will sometimes write $a\geqslant 0$ or $a > 0$ for $a$ positive semi-definite or positive definite, resp.).

Recall that a linear map $\phi : \matr{n}\to\matr{m}$ is called \emph{positive} if $\phi(\matr{n}^+)\subset\matr{m}^+$. Further, $\phi$ is called \emph{completely positive} (CP) if $\id{}\otimes\phi$ is a positive map on algebra $\matr{n}(\matr{n})\simeq \matr{n}\otimes\matr{n}$. Likewise, $\phi$ is called \emph{completely copositive} (coCP) if $\theta \circ \phi$, with $\theta$ being a \emph{transposition}, is CP. If $\phi$ is both CP and coCP it is called a \emph{PPT map}. Sets of all CP, coCP and PPT maps are pointed convex cones in $B(\matr{n},\matr{m})$, closed with respect to supremum norm topology. When $n=m=d$, which is the case in present framework, we will denote them $\cpe{\matrd}$, $\cocpe{\matrd}$ and $\ppte{\matrd}$, respectively. By Choi's theorem \cite{Choi1975}, $\phi$ is CP if and only if a matrix
\begin{equation}
    \choi{\phi} = \sum_{i,j=1}^{n} E_{ij} \otimes \phi(E_{ij}) = [\phi(E_{ij})],
\end{equation}
where $E_{ij}$ are matrix units spanning $\matr{n}$, is positive semi-definite in $\matr{nm}$. $\choi{\phi}$ is called the \emph{Choi's matrix of $\phi$} and a mapping $\phi\mapsto\choi{\phi}$ is a bijection from $B(\matr{n},\matr{m})$ to $\matr{n}\otimes\matr{m}\simeq\matr{nm}$, called the \emph{Choi-Jamiołkowski isomorphism} \cite{Jamiolkowski1972,Choi1975}. Then, $\phi$ is coCP iff $\choi{\theta\circ\phi} = \choi{\phi}^{\ptrans} \in \matr{nm}^+$, with $\ptrans$ denoting the partial transposition with respect to the second factor, $(a\otimes b)^{\ptrans} = a \otimes b^\transpose$. In consequence, $\phi$ is PPT iff both $\choi{\phi},\choi{\phi}^{\ptrans}\in\matr{nm}^+$, i.e.~when $\choi{\phi}$ is a so-called \emph{PPT matrix}.

We will grant a primary attention to an important subclass of PPT maps, so-called \emph{entanglement breaking maps}. A map $\phi\in\ppte{\matrd}$ is called \emph{entanglement breaking} (EB) iff $\choi{\phi}$ is a \emph{separable} matrix \cite{Rahaman2018,Horodecki2003}, i.e.~when there exist sets $\{A_i\},\{B_i\}\subset\matrd^+$ such that
\begin{equation}
    \choi{\phi} = \sum_i A_i \otimes B_i .
\end{equation}
Equivalently, $\phi$ is entanglement breaking iff $(\id{}\otimes\phi)(X)$ is always separable, even for entangled $X \in \matrd\otimes\matrd$; we will denote a cone of all EB maps by $\ebe{\matrd}$. An important property, which will be employed by us in some proofs, is that both sets $\ebe{\matrd}$ and $\ppte{\matrd}$ are examples of \emph{mapping cones} \cite{Devendra2022,Girard2021}, i.e.~closed convex cones, invariant with respect to compositions with CP maps, from left and from right: for any map $\phi$ which is EB (PPT) and any two CP maps $\psi_1$, $\psi_2$ it holds that $\psi_1 \circ \phi \circ \psi_2$ is also EB (PPT).

\section{Asymptotic maps of evolution families}
\label{sec:Asymptotics}

The concept of positive map became a crucial ingredient in modeling evolution in mathematical theory of quantum mechanics. Recall that state of a system at any time $t\geqslant 0$ is expressed as a \emph{density operator} $\rho_t$, i.e.~a time-dependent, positive semi-definite trace class operator of trace (norm) $1$, acting on some (possibly infinite dimensional) Hilbert space. Given $\rho_0$, the state at later times is then $\rho_t = \Lambda_t (\rho_0)$, where we require $\Lambda_t$ to be a \emph{positive map} and also \emph{trace preserving} (TP), at least on Banach algebra of trace class operators, so that $\tr{\Lambda_t}(\rho) = \tr{\rho}$. These two properties (positivity and trace preservation) guarantee that $\rho_t$ will not lose its statistical interpretation as a mixed state, i.e.~it remains a density operator. Naturally, a family $(\Lambda_t)_{t\in\reals_+}$ of all such maps then encodes a quantum-mechanical evolution and as such is commonly called the \emph{quantum dynamical map} or \emph{quantum evolution family} \cite{RiHu2012,Chruscinski2022}.

Throughout this paper, we will be interested in global asymptotic behavior of such families acting on $\matrd$ and, in particular, whether they either become EB or remain arbitrarily close to set of EB maps. Recall that for a metric space $(X,d)$, the \emph{distance of element $x\in X$ to subset $A\subset X$} is defined as
\begin{equation}
    d(x,A) = \inf_{a\in A}{d(x,a)}.
\end{equation}
To mirror the aforementioned asymptotic behavior, we introduce the following:
\begin{definition}
    A dynamical map $(\phi_t)_{t \in\reals_+} $ on $\matrd$ will be called:
    \begin{enumerate}
        \item \emph{asymptotically entanglement breaking} (asymptotically EB) if $\phi_t$ approaches the cone $\ebe{\matrd}$ asymptotically, i.e. $\lim_{t\to\infty}{d(\phi_t, \ebe{\matrd})} = 0$ with metric given by supremum norm, $d(\phi,\psi) = \supnorm{\phi-\psi}$;
        \item \emph{eventually entanglement breaking} (eventually EB) if it actually reaches $\ebe{\matrd}$ in finite time, i.e.~when there exists $t_0 > 0$ s.t.~$\phi_t \in \ebe{\matrd}$ for all $t\geqslant t_0$;
        \item \emph{asymptotically PPT} if $\phi_t$ approaches the cone $\ppte{\matrd}$ asymptotically;
        \item \emph{eventually PPT} if it actually reaches $\ppte{\matrd}$ in finite time.
					
    \end{enumerate}    
\end{definition}

It may happen -- as is the case in our analysis -- that evolution families governed by time-dependent generators may tend not merely to a fixed positive map, but rather to some map-valued function. Let $f,g : \reals\to (X,d)$ for again $(X,d)$ a metric space. We will say that $f$ \emph{asymptotically tends to} $g$, $f\xrightarrow{a}g$, if $\lim_{x\to\infty}d(f(x),g(x)) = 0$.

In what follows we assume that maps we are investigating are diagonalizable (a set of diagonalizable maps is dense in a set of all linear maps). A linear map $\phi \in B(\matrd)$ is said to be \emph{diagonalizable} if
\begin{equation}\label{eq:DiagMap}
    \phi = \sum_{i=1}^{d^2} \lambda_i \proj{i} ,
\end{equation}
where $\{\lambda_i\}_{i=1}^{d^2}$ are (in general complex) eigenvalues and $\proj{i}$ are rank one projection operators onto eigenvectors of $\phi$, satisfying $\proj{i}\proj{j} = \delta_{ij} \proj{i}$ (we count multiplicities). If $\phi$ is Hermiticity preserving, then its spectrum is symmetric w.r.t. real line, that is, eigenvalues are either real or come in pairs $(\lambda,\,\overline{\lambda})$. Denote by $\phi^\dual$ the adjoint of $\phi$ with respect to standard Hilbert-Schmidt inner product $\hsiprod{a}{b} = \tr{a^\hadj b}$, so that it satisfies $\hsiprod{a}{\phi(b)} = \hsiprod{\phi^\dual(a)}{b}$. Diagonalizabity of $\phi$ infers existence of an biorthogonal system $\{(X_i)_{i=1}^{d^2},(Y_i)_{i=1}^{d^2}\}$ in $\matrd$ consisting of eigenvectors (unnormalized in general) of $\phi$ and $\phi^\dual$ such that $\hsiprod{X_i}{Y_j} = \delta_{ij}$ and
\begin{equation}
    \phi(X_i) = \lambda_i X_i , \quad \phi^\dual(Y_i) = \overline{\lambda_i} Y_i .
\end{equation}
Then, one has $\proj{i} = \hsiprod{Y_i}{\cdot}X_i$.

Recall that $\phi$ is trace preserving if and only if $\phi^\dual$ is unital, i.e.~$\phi^\dual (I) = I$. In this case one eigenvalue, traditionally $\lambda_1$, equals to 1 and $Y_1 = I$, $\tr{X_1} = 1$ in consequence. If moreover $\phi$ is positive, then due to the celebrated Perron-Frobenius theorem \cite{Stoermer1963,Paulsen2003} one has $|\lambda_i|\leqslant 1$ and $X_1 \geqslant 0$. Hence, for quantum evolution family $(\phi_t)_{t\in\reals_+}$, i.e.~a family of of trace preserving and positive maps, we have 
\begin{equation}
    \lambda_1 (t) = 1, \quad \proj{1}(t) (\rho) = (\tr{\rho})\,\Omega(t), \quad 
\end{equation}
for some time-dependent matrix $\Omega(t) \in\matrd^+$ s.t.~$\tr{\Omega(t)} = 1$. Let us therefore start with the following simple observation.

\begin{theorem}\label{thm:eEBSufficientCondition}
    Let $(\phi_t)_{t\in\reals_+}$ be a quantum evolution family on $\matrd$ which is continuous, and let $\spec{\phi_t} = \{\lambda_i (t)\}_{i=1}^{d^2}$ (counting multiplicities). Assume
    \begin{enumerate}
        \item $\lim_{t\to\infty}{\lambda_i (t)} = 0$ for $i \geqslant 2$,
        \item projection $\proj{1}(t) = (\tr{\cdot})\,\Omega(t)$
        asymptotically tends to a projection $\mathsf{Z}_t$ defined via
        \begin{equation}
            \mathsf{Z}_t (\rho) = (\tr{\rho})\,\omega(t)
        \end{equation}
        for some $\omega(t) \in\matrd^+$, $\tr{\omega}(t) = 1$ (in particular, $\omega(t)$ may be constant),
        \item spectrum of $\omega(t)$ is uniformly separated from $0$, i.e.~there exists $\epsilon > 0$ s.t.~$\min{\spec{\omega(t)}} \geqslant \epsilon$ for all $t\in\reals_+$.
    \end{enumerate}
    Then, $(\phi_t)_{t\in\reals_+}$ is eventually EB.
\end{theorem}

\begin{proof}
    The proof is quite straightforward. We can represent each map $\phi_t$ in form of its spectral decomposition
    \begin{equation}
        \phi_t = \proj{1}(t) + \sum_{i\geqslant 2}\lambda_i (t) \proj{i} (t),
    \end{equation}
    where $\proj{1}(t) = (\tr{\cdot})\,\Omega(t)$ for some $\Omega(t)\geqslant 0$, $\tr{\Omega(t)} = 1$, $|\lambda_i (t)|\leqslant 1$. From the assumptions we clearly see the whole family $(\phi_t)_{t\in\reals_+}$ asymptotically tends to $\mathsf{Z}_t$ w.r.t.~supremum norm in Banach space $B(\matrd)$: indeed, notice
    \begin{equation}
        \supnorm{\phi_t - \mathsf{Z}_t} \leqslant \supnorm{\proj{1} (t) - \mathsf{Z}_t} + \sum_{i\geqslant 2}|\lambda_i (t)|
    \end{equation}
    after employing triangle inequality and $\supnorm{\proj{i}(t)} = 1$, and the RHS tends to $0$ as $t\to\infty$. Let there exist a family $\{\omega(t) : t\in\reals_+\}$ of positive definite matrices of trace $1$ in $\matrd$. By Lemma \ref{lemma:PomegaProperties} (in Appendix \ref{app:MathematicalSupplement}) each map $\mathsf{Z}_t = (\tr{\cdot})\,\omega(t)$ lays in strict interior of $\ebe{\matrd}$ and there exists a family $\{\mathcal{B}(\mathsf{Z}_t ,r(t)) : t\in\reals_+\}$ of open balls of radii $r(t)$, each one centered at $\mathsf{Z}_t$ and contained in $\operatorname{Int}{\ebe{\matrd}}$. Now, since spectra of $\omega(t)$ were uniformly separated from $0$, no sequence of matrices $(\omega(t_n))_{n\in\naturals}$ lays arbitrarily close to the boundary of $\matrd^+$ and family $\{\mathsf{Z}_t : t\in\reals_+\}$ is separated from a boundary of $\ebe{\matrd}$ in consequence. This means that there exists a \emph{minimal} radius $r_0 = \inf{\{r(t) : t\in\reals_+\}} > 0$ s.t.~a family of open balls $\{\mathcal{B}(\mathsf{Z}_t,r_0) : t\in\reals_+\}$ is fully contained in $\operatorname{Int}{\ebe{\matrd}}$; denote $\mathcal{U} = \bigcup_{t\geqslant 0}\mathcal{B}(\mathsf{Z}_t,r_0)$. If now $\supnorm{\phi_t - \mathsf{Z}_t} \to 0$ then there exists $t_0 \in \reals_+$ that for $t \geqslant t_0$ we have $\supnorm{\phi_t - \mathsf{Z}_t} <  r_0$, that is $\phi_t \in \mathcal{U}$ and the family $(\phi_t)_{t\in\reals_+}$ is eventually EB as claimed.
\end{proof}

Since we are interested in quantum dynamics, we grant a special attention to families $(\Lambda_t)_{t\in\reals_+}$ of maps subject to some form of time-local Master Equation
\begin{equation}
    \dot{\Lambda}_t = L_t \circ \Lambda_t
\end{equation}
for a time-dependent \emph{generator} $L_t$. It comes with no surprise that in certain simplified cases, when a closed form of $(\Lambda_t)_{t\in\reals_+}$ may be computed, asymptotic properties of $\Lambda_t$ may be entirely deduced from spectral properties of the generator:

\begin{theorem}\label{thm:TheTheorem}
Let $(L_t)_{t\in\reals_+}$ be a family of diagonalizable maps on $\matrd$, with $t \mapsto L_t$ at least piecewise continuous and commutative, $L_t \circ L_s = L_s \circ L_t$ for all $t,s \in\reals_+$. Assume that, for all $t\in\reals_+$,
\begin{enumerate}
    \item\label{TheTheoremAssum1} $0\in\spec{L_t}$ is of multiplicity $1$,
    \item\label{TheTheoremAssumUnbounded} for $\mu(t)\in\spec{L_t}\setminus \{0\}$ we have
    \begin{equation}
        \lim_{t\to\infty}{\int\limits_{0}^{t}\Re{\mu (s)} ds} = -\infty ,
    \end{equation}
    \item\label{TheTheoremAssum2} $\ker{L_t} = \complexes \omega$ for a constant, positive definite matrix $\omega$.
\end{enumerate}
Then, a family $(\Lambda_t)_{t\in\reals_+}$ generated by $L_t$ is eventually EB.
\end{theorem}

\begin{proof}
Let $L_t$ be a map on $\matrd$ satisfying the assumptions and let a family $(\Lambda_t)_{t\in\reals_+}$ solve the Cauchy problem of a form
\begin{equation}
    \dot{\Lambda}_t = L_t \circ \Lambda_t , \quad \Lambda_0 = \id{}.
\end{equation}
The ray $\ker{L_t} = \complexes\omega$ is a time-invariant eigenspace of $\Lambda_t$ for eigenvalue 1. Indeed, note that by commutativity assumption we have
\begin{equation}
    \Lambda_t = \exp{\int\limits_{0}^{t}L_s ds} = \id{} + \sum_{n=1}^{\infty}\frac{1}{n!}\left(\int\limits_{0}^{t}L_s ds\right)^n ,
\end{equation}
converging uniformly; then, $L_t(\omega) = 0$ simply yields $\Lambda_t (\omega) = \omega$. Let $\spec{L_t} = \{\mu_i (t)\}_{i=1}^{d^2}$ (including multiplicities), such that $\mu_1 (t) = 0$. One can apply a spectral decomposition of $L_t$,
\begin{equation}
    L_t = \sum_{i=1}^{d^2} \mu_i (t) \proj{i} = 0 \cdot \proj{\omega} + \sum_{i > 1} \mu_i (t) \proj{i}(t),
\end{equation}
where $\proj{i}(t)$ denote rank-one projection operators onto distinct eigenspaces of $L_t$ (time-dependent in general), satisfying $\proj{i}(t)\proj{j}(t) = \delta_{ij} \proj{j}(t)$; here we set $\proj{1}(t) = \proj{\omega}$, the projection onto $\complexes\omega$ (we leave the $0$ eigenvalue present for clarity). Therefore, $\Lambda_t$ shares the same eigenspaces and its eigenvalues $\lambda_i (t)$ are
\begin{equation}\label{eq:SpectLambdat}
    \lambda_i(t) = \exp{\int\limits_{0}^{t}\mu_i (s) ds}
\end{equation}
with $\lambda_1 (t) = 1$. Now we can expand $\Lambda_t$ into its spectral decomposition
\begin{equation}
    \Lambda_t = \proj{\omega} + \sum_{i \geqslant 2} \lambda_i (t) \proj{i}(t),
\end{equation}
where by \eqref{eq:SpectLambdat} we have
\begin{equation}
    \lambda_j (t) = e^{R_j(t)}e^{i K_j (t)},
\end{equation}
with $R_j (t)$ and $K_j (t)$ respectively standing for real and imaginary parts of $\int_0^t \mu_j (s)\,ds$. By assumption, all functions $R_j (t)$ are unbounded from below and so $|\lambda_j (t)| \to 0$ as $t\to\infty$. By Lemma \ref{lemma:PomegaProperties}, $\proj{\omega}\subset \operatorname{Int}{\ebe{\matrd}}$ whenever $\omega$ is strictly positive definite. Hence, family $(\Lambda_t)_{t\in\reals_+}$ satisfies all assumptions of Theorem \ref{thm:eEBSufficientCondition} (with projection $\mathsf{Z}_t$ constant and equal to $\proj{\omega}$, i.e.~$\Omega(t) = \omega(t) = \omega$) and the claim follows.
\end{proof}

As a special case of the above, we formulate a following theorem applicable to semigroups of maps.

\begin{theorem}\label{thm:TheTheoremSemigroups}
    Let $(\phi_t)_{t\in\reals_+}$ be a semigroup, $\phi_t = e^{tL}$. Assume that
    \begin{enumerate}
        \item $0\in \ker{L}$ is of multiplicity $1$,
        \item for $\mu \in \spec{L}\setminus\{0\}$ we have $\Re{\mu} < 0$, i.e.~a non-zero part of spectrum of the generator lays on left complex half-plane,
        \item $\ker{L} = \complexes \omega$ for positive definite matrix $\omega$.
    \end{enumerate}
    Then, family $(\phi_t)_{t\in\reals_+}$ is eventually EB.
\end{theorem}

\begin{proof}
    By condition $\Re{\mu} < 0$ wee see that $\int_0^t \Re{\mu}\,ds = t\Re{\mu} \to -\infty$, so all assumptions of Theorem \ref{thm:TheTheorem} are met.
\end{proof}

\section{Eventual EB-divisibility}
\label{sec:EBdivisibility}

In this section we define a notion of eventual EB-divisibility of families of maps and provide some insight into intrinsic relations between their asymptotic behavior and eEB-divisibility.

Let us recall that a dynamical map $(\phi_t)_{t\in\reals_+}$ is \emph{divisible} if for any $t \geqslant  s \geqslant 0$ one has $\phi_t = V_{t,s}\circ\phi_s$ for $V_{t,s}$ called a \emph{propagator}. It is clear that any invertible dynamical map is divisible and the corresponding propagator reads $ V_{t,s} = \phi_t\circ\phi_{s}^{-1}$. A map $(\phi_t)_{t\in\reals_+}$ is
\begin{itemize}
    \item P-divisible if $V_{t,s}$ is positive and trace preserving,
     \item CP-divisible if $V_{t,s}$ is completely positive and trace preserving.
\end{itemize}
P- and CP-divisibility are of special significance for quantum theory and are studied extensively up to this day in various contexts in physics literature. We draw the reader's attention also to recently introduced notion of \emph{D-divisibility} \cite{Szczygielski_2023}, where the propagators $V_{t,s}$ are not completely positive, but rather \emph{decomposable} maps and thus provide a new class of weakly non-Markovian, P-divisible dynamics. 

In general, properties of propagators of the dynamical maps are mirrored by those of the dynamics itself in a sense that, say, CP-divisibility (P-div., D-div., resp.) is a sufficient condition for complete positivity (positivity, decomposability, resp.) of the dynamics, but not a necessary one in general. Therefore it seems reasonable to expect a roughly the same asymptotic behavior both from the dynamical map and from its propagator at larger times. This infers that it is justified to \emph{define} a new type of divisibility, the \emph{eventual EB-divisibility}, by requiring the propagators to become entanglement breaking at large times, i.e.~to be eventually EB as introduced earlier.

The notion of eventual divisibility is by no means limited to entanglement breaking case, however:
 let $\mathscr{X}\subset B(\matrd)$ denote a subset of trace preserving linear maps on $\matrd$; we define the \emph{eventual $\mathscr{X}$-divisibility} of a family of maps in the following way:
 \begin{definition}\label{def:Xdivisibility}
    A family of trace preserving linear maps $(\phi_t)_{t\in\reals_+}$ on $\matrd$ will be called \emph{eventually $\mathscr{X}$-divisible} (e$\mathscr{X}$-divisible) iff it is divisible and every family of propagators $(V_{t,s})_{t\geqslant s}$, $s\in\reals_+$, eventually lays in $\mathscr{X}$, i.e.
    \begin{equation}
    \forall \, s \in\reals_+ \, \exists \, \Delta(s) > s \, \forall \, t\geqslant \Delta(s) : V_{t,s}\in \mathscr{X}.
    \end{equation}
\end{definition}

Such definition can provide a significant generalizations of known types of divisibility, since we no longer demand from propagators to be always ``of some type'', but only after a nonzero time. In particular, we can define a whole hierarchy of variations of eventual divisibility:

\begin{definition}\label{def:PPTdivisibility}

A divisible dynamical map $(\phi_t)_{t\in\reals_+}$  will be called 

\begin{itemize}
    \item \emph{eventually P-divisible} (eP-divisible) if $(V_{t,s})_{t\geqslant s}$ is eventually positive for every $s\geqslant 0$,

 \item \emph{eventually CP-divisible} (eCP-divisible) if $(V_{t,s})_{t\geqslant s}$ is eventually completely positive for every $s\geqslant 0$,

  \item \emph{eventually PPT-divisible} (ePPT-divisible) if $(V_{t,s})_{t\geqslant s}$ is eventually PPT for every $s\geqslant 0$,

  \item and finally \emph{eventually EB-divisible} (eEB-divisible) if $(V_{t,s})_{t\geqslant s}$ is eventually EB for every $s\geqslant 0$.
\end{itemize}
\end{definition}

As an immediate consequence, we see that all e$\mathscr{X}$-divisible families tend to belong to $\mathscr{X}$ itself -- clearly, e$\mathscr{X}$-divisibility is \emph{sufficient} for a family to be eventually in $\mathscr{X}$. It turns out that it is not \emph{necessary} though, as we show in one of the examples further below.

\begin{theorem}\label{thm:eX}
 If a family $(\phi_t)_{t\in\reals_+}$ is e$\mathscr{X}$-divisible then it is eventually in $\mathscr{X}$.
\end{theorem}

\begin{proof}
    Eventual $\mathscr{X}$-divisibility states that $V_{t,0}\in \mathscr{X}$ whenever $t \geqslant \Delta(0)$. Since $\phi_t = V_{t,0}$, one simply takes $t_0 = \Delta(0) > 0$ and the claim follows.
\end{proof}

It is clear that in order for a propagator to become EB it must become PPT first, so for aforementioned types of divisibility the following chain of implications holds:
$$ \text{eEB-div.} \, \Rightarrow \, \text{ePPT-div.} \, \Rightarrow \, \text{eCP-div.}  \, \Rightarrow \, \text{eP-div.} $$

Naturally, in current framework we restrict our attention solely to eEB-divisible evolution families, leaving the remaining types of divisibility as an interesting direction of further research. Theorem \ref{thm:eEB} yields the following

\begin{corollary}\label{thm:eEB}
 If a family $(\phi_t)_{t\in\reals_+}$ is eEB-divisible then it is eventually EB.
\end{corollary}

\begin{remark}
    The eEB-divisibility cannot be a straightforward restatement of, say, P-divisibility, i.e.~one may not simply demand from $V_{t,s}$ to \emph{be} EB. Clearly, from continuity of function $t\mapsto V_{t,s}$, for small differences $t-s$ maps $V_{t,s}$ lay, informally speaking, in a small neighborhood of $V_{s,s} = \id{}$ which is trivially not PPT and therefore not entanglement breaking. This means there will always exist some nonempty interval $[s,t_0]$ such that $V_{t,s}$ would not be EB for any $t\in[s,t_0]$ and the EB condition may be only satisfied for $t-s$ large enough, hence the condition $\Delta(s) > 0$ as we stated.
\end{remark} 

A simple following result applies:

\begin{theorem}\label{thm:TheTheoremGeneral}
    Let $(\phi_t)_{t\in\reals_+}$ be a family of divisible, invertible and diagonalizable linear maps on $\matrd$. Let $\spec{\phi_t} = \{\lambda_i (t)\}_{i=1}^{d^2}$ (counting multiplicities). Assume that
    \begin{enumerate}
        \item $\lambda_1 (t) = 1$,
        \item $\lim_{t\to\infty}{\lambda_i (t)} = 0$ for $i \geqslant 2$,
        \item projection $\proj{1}$ is independent of $t$.
    \end{enumerate}
    Then, the following statements hold:
    \begin{enumerate}
        \item\label{thm:TheTheoremGeneral:1} If $\proj{1}\in\operatorname{Int}{\ebe{\matrd}}$ then $(\phi_t)_{t\in\reals_+}$ is eventually EB and eEB-divisible.
        \item\label{thm:TheTheoremGeneral:2} If $\proj{1}\notin\ebe{\matrd}$ then $(\phi_t)_{t\in\reals_+}$ is neither eventually EB nor eEB-divisible.
    \end{enumerate}
     
\end{theorem}

\begin{proof}
    Family $(\phi_t)_{t\in\reals_+}$ is eventually EB directly from Theorem \ref{thm:eEBSufficientCondition}. By invertibility assumption, $V_{t,s} = \phi_t \circ \phi_{s}^{-1}$. Map $\phi_{s}^{-1}$, being a holomorphic function of $\phi_s$, has the same eigenspaces and therefore
    \begin{equation}
        V_{t,s} = \proj{1} + \sum_{i,j \geqslant 2} \zeta_{ij}(t,s) \proj{i}(t)\proj{j}(s)
    \end{equation}
    after simple check, where we defined $\zeta_{ij}(t,s) = \lambda_i (t) / \lambda_j (s)$ for brevity. Then $\lambda_i (t) \to 0$ implies $\zeta_{ij}(t,s) \to 0$ and thus $V_{\cdot,s}\xrightarrow{a}\proj{1}$ for all $s$, namely $\supnorm{V_{t,s}-\proj{1}}\to 0$ as $t\to\infty$. Then, if $\proj{1}\in\operatorname{Int}{\ebe{\matrd}}$ there exists an open neighborhood $\mathcal{U}_s$ of $\proj{1}$ contained inside $\ebe{\matrd}$ and some $t_0 > s$ large enough such that $V_{t,s} \in \mathcal{U}_s$ for all $t\geqslant t_0$, i.e.~$V_{t,s}$ becomes EB. Similarly, when $\proj{1}$ lays in the complement of $\ebe{\matrd}$, so does $\mathcal{U}$ (since the complement is open) and neither $\phi_t$ nor $V_{t,s}$ become EB.
\end{proof}

\begin{theorem}\label{thm:LtCommutativeEBdiv}
    Family $(\Lambda_t)_{t\in\reals_+}$ characterized in Theorem \ref{thm:TheTheorem} is eEB-divisible.
\end{theorem}

\begin{proof}
    As $\proj{1} = \proj{\omega}$ lays in the interior of $\ebe{\matrd}$ by Lemma \ref{lemma:PomegaProperties}, eEB-divisibility of this family is a consequence of Theorem \ref{thm:TheTheoremGeneral}.
\end{proof}

The following theorem shows that in a simple case of semigroups, notions of being eventually EB and eEB-divisible are totally equivalent. This situation is analogous to other forms of X-divisibility, where X denotes P, D, or CP.

\begin{theorem}\label{thm:SemigroupEB}
    A semigroup $(\phi_t)_{t\in\reals_+}$ of maps on $\matrd$ is eEB-divisible if and only if it is eventually EB.
\end{theorem}

\begin{proof}
    Direction ``$\Rightarrow$'' follows immediately from Theorem \ref{thm:eEB}. For the opposite note that if $\phi_t$ is EB for $t\geqslant t_0$, then a propagator $V_{t,s}=\phi_{t-s}$ is EB for any $t\geqslant \Delta(s) = s+t_0 > s$, i.e.~$(\phi_t)_{t\in\reals_+}$ is eEB-divisible.
\end{proof}

\section{CP-divisible dynamics}
\label{sec:CPdivisible}

In this section we elaborate on probably the most distinguished and well-studied case of quantum evolution families, namely the CP-divisible dynamical maps governed by infinitesimal generators in the celebrated Gorini-Kossakowski-Lindblad-Sudarshan (GKLS) form (see e.g.~references \cite{Chruscinski2017,Alicki2006a,Breuer2002} for an excellent overview of the subject). The main result is that, roughly speaking. when the dynamics possesses a unique stationary state which is strictly positive definite then it becomes entanglement breaking. We show it rigorously in two related cases, namely for commuting GKLS generators and for quantum dynamical semigroups. We conjecture, however, that this observation applies to a much broader class of more general, non-commuting generators.

\begin{theorem}\label{thm:LGKS}
    Let $(L_t)_{t\in\reals_+}$ be a commutative family of time-dependent GKLS generators and assume $\spec{L_t} = \{\mu_i (t)\}_{i=1}^{d^2}$ (counting multiplicities), $\mu_1 (t) = 0$, $\ker{L_t} = \complexes \omega$. If
    \begin{equation}\label{eq:LGKSassump}
        \lim_{t\to\infty}\int\limits_{0}^{t}\Re{\mu_i (s)}ds = -\infty
    \end{equation}
    for $i\geqslant 2$ and $\omega > 0$ then family $(\Lambda_t)_{t\in\reals_+}$ generated by $L_t$ is eEB-divisible and eventually EB.
\end{theorem}

\begin{proof}
    $L_t$, being in GKLS form, nullifies the trace, $\tr{L_t (\rho) = 0}$, and so $0 \in \spec{L_t}$. The remaining part of spectrum lays on the complex left half-plane (and is symmetric w.r.t.~real axis) and satisfies condition \eqref{eq:LGKSassump} by assumption, so Theorem \ref{thm:LtCommutativeEBdiv} applies and the proof is finished.
\end{proof}

In particular, assumptions of Theorem \ref{thm:LGKS} may be satisfied in the prominent case of quantum dynamical semigroups:

\begin{theorem}\label{thm:LGKSsemigroup}
    Let $L$ be a GKLS generator and let $\omega$ be its unique stationary state. If $\omega > 0$ then a semigroup $(e^{tL})_{t\in\reals_+}$ is eEB-divisible and eventually EB.
\end{theorem}

\begin{proof}
    We see that since all non-zero eigenvalues $\mu$ of $L$ lay in the left complex half-plane, all integrals $\int_0^t \Re{\mu} \, ds = t \Re{\mu}$ are unbounded from below and Theorems \ref{thm:SemigroupEB} and \ref{thm:LGKS} apply.
\end{proof}

Finally, the following result shows that in case of CP-divisible dynamics it is enough for the propagator to be entanglement breaking at \emph{one} instant, in order to be eventually EB:

\begin{theorem}\label{thm:CPdivPPTdiv}
    Let a family $(\Lambda_t)_{t\in\reals_+}$ be CP-divisible. If it happens that for every $s\geqslant 0$ a propagator $V_{t_0 (s),s}\in\ebe{\matrd}$ for \emph{some} $t_0 (s) > s$, then also $V_{t,s}\in\ebe{\matrd}$ for \emph{all} $t \geqslant t_0 (s)$, i.e.~family $(\Lambda_t)_{t\in\reals_+}$ is also eEB-divisible.
\end{theorem}

\begin{proof}
    Let $s\geqslant 0$ and let $t_0 (s) > s$ be such that $V_{t_0 (s),s}\in\ebe{\matrd}$. Notice, that for any $t > t_0 (s)$ the CP-divisibility guarantees that we have
    \begin{equation}
        V_{t,s} = V_{t,t_0(s)} V_{t_0(s),s}
    \end{equation}
    where $V_{t,t_0(s)}\in\cpe{\matrd}$ and $V_{t_0(s),s}\in\ebe{\matrd}$; then $V_{t,s}$ is also EB as a composition from mapping cone property of $\ebe{\matrd}$.
\end{proof}

\section{Semigroups and \texorpdfstring{PPT\textsuperscript{2}-conjecture}{PPT-squared-conjecture}}
\label{sec:PPT2}

Here we make few remarks on some correlations between eventually EB semigroups and the famous PPT\textsuperscript{2}-conjecture. Let us recall that it was conjectured by Christandl that a composition of any two PPT maps is always entanglement breaking \cite{Banff}. Up to now, the PPT\textsuperscript{2}-conjecture was rigorously proved in trivial case of algebra $\matr{2}$ (where it results basically from Peres-Horodecki criterion of separability) and also for some specific classes of maps beyond dimension $2$ (see e.g.~\cite{Singh2022}). An interesting result appeared in \cite[Thm. 3.5]{Kennedy2018}, where it was shown that the conjecture holds in \emph{asymptotic} sense: for every PPT map $\phi$ which is trace preserving or unital, the sequence $(\phi^n)_{n\in\naturals}$ of iterative compositions of $\phi$ tends to be arbitrarily close to the set of all entanglement breaking maps,
\begin{equation}\label{eq:AsymptEBcondition}
    \lim_{n\to\infty}d(\phi^n , \ebe{\matrd}) = 0.
\end{equation}
Moreover, this result was in a sense refined in \cite[Thm. 4.4]{Rahaman2018} where it was shown that for every unital and trace preserving (bistochastic) PPT map $\phi$ there exists a finite $k \in \naturals$ s.t.~$\phi^k$ is actually entanglement breaking. The asymptotic result \eqref{eq:AsymptEBcondition} of \cite{Kennedy2018} allows to formulate an interesting sufficient condition for semigroups to be asymptotically EB, even with no \emph{a priori} knowledge of spectral properties of its stationary state. We formulate it in form of Theorems \ref{thm:eEBsemigroupCP} and \ref{thm:AsympEBsemigroup} below, where in the former we restrict attention to the \emph{completely positive} case, while the latter is a generalization concerning any semigroup which is unital or trace preserving.

\begin{theorem}\label{thm:eEBsemigroupCP}
    Let $(\phi_t)_{t\in\reals_+}$ be a semigroup of completely positive, unital and trace preserving linear maps on $\matrd$. If there exists $s > 0$ such that $\phi_{s}$ is PPT then the semigroup is eventually EB.
\end{theorem}

\begin{proof}
    Let $\phi_s \in \ppte{\matrd}$. The result of \cite[Thm. 4.4]{Rahaman2018} yields existence of some $k\in\naturals$ s.t.~$\phi_{s}^{k} = \phi_{ks}$ is EB. Then, for any $t \geqslant ks$ we have $\phi_{t} = \phi_{t-ks}\circ\phi_{ks}$ where $\phi_{t-ks}$ is CP. Hence, $\phi_t$ is EB by mapping cone property of $\ebe{\matrd}$.
\end{proof}

We can actually relax the complete positivity requirement in order to treat more general class of semigroups. By contrast to Theorem \ref{thm:eEBsemigroupCP} where we had CP-divisibility in our disposal to conclude on asymptotical behavior, in the result below we make a little stronger assumption about semigroup being PPT not at one point, but over some interval:
\begin{theorem}\label{thm:AsympEBsemigroup}
    Let $(\phi_t)_{t\in\reals_+}$ be a semigroup of unital or trace preserving linear maps on $\matrd$. If there exists $s > 0$ such that $\phi_{s}$ lays inside the cone $\ppte{\matrd}$, then the semigroup is asymptotically EB.
\end{theorem}

\begin{proof}
    If $\phi_s$ is an interior point in $\ppte{\matrd}$, it is separated from the cone's boundary. Therefore, by continuity of $t\mapsto\phi_t$, there must exist a non-empty interval $\mathcal{I}_1=[t_1, t_2] \subset (0,\infty)$ s.t.~all maps $\phi_t$ are PPT for $t\in\mathcal{I}_1$. From mapping cone property of PPT maps we see that also $\phi_{t}^{n} = \phi_{nt}$ are PPT for any $t\in\mathcal{I}_1$, $n\in\naturals$. Let us define a family $\{\mathcal{I}_n : n\in\naturals\}$ of shifted and scaled intervals, $\mathcal{I}_n = [nt_1, nt_2]$. Then, for $t\in\mathcal{I}_1$ we have $nt\in\mathcal{I}_n$ and so a family $\{\phi_{t} : t\in\mathcal{I}_n\}$ is PPT, for all $n$. In consequence, $\{\phi_t : t\in\bigcup_{n\in\naturals}\mathcal{I}_n\}$ is also PPT. We will make use of a following simple lemma (for proof, see Lemma \ref{lemma:Halfline} in Appendix \ref{app:MathematicalSupplement}):
    
    \begin{lemma}
        There exists a half-line $[t_*,\infty)$, $t_*>0$, s.t.~a family $\{\mathcal{I}_n : n>n_0\}$ is its covering for $n_0 \in\naturals$ large enough.
    \end{lemma}
    
    What this lemma gives is that there exists some $t_* > 0$ s.t.~$\phi_t$ is PPT for all $t\geqslant t_*$ and there are no ``holes'' where a PPT property may be suddenly lost. Let us then define a new family of maps $(\psi_t)_{t\in\reals_+}$ by shifting the origin to point $t_*$, i.e.~$\psi_t = \phi_{t+t_*}$. For brevity, denote $f(t) = d(\phi_t, \ebe{\matrd})$ and $f_* = f(\cdot + t_*)$. Naturally, $f_*$ is then a distance between $\psi_t$ and $\ebe{\matrd}$ and both $f$, $f_*$ are continuous on their respective domains. Now, take any $t \geqslant 0$ and consider a sequence $(\psi_{nt})=(\psi_{t}^{n})$. Map $\psi_{nt} = \phi_{n(t+t_*)}$ is clearly PPT for all $n\in\naturals$ so condition \eqref{eq:AsymptEBcondition} of \cite[Thm. 3.5]{Kennedy2018} yields $(\psi_{nt})$ is asymptotically EB in the sense that $f_* (nt)\to 0$ as $n\to\infty$. Since this is true for all $t > 0$, application of \emph{Croft's lemma} \cite{Kingman1963} indicates also $\lim_{t\to\infty}f_* (t) = 0$, i.e.~a family $(\phi_t)_{t\in\reals_+}$ is asymptotically EB.
    \end{proof}
    
\begin{remark}
    We note that checking if the semigroup actually enters the interior of $\ppte{\matrd}$ reduces to finding at least one instant $t>0$, for which spectra of both Choi's matrices $\choi{\phi_t}$ and $\choi{\phi_t}^{\ptrans}$ are strictly positive, and as such is potentially easily achievable, at least by numerical means. 
\end{remark}

We conclude this section with a simple observation about semigroups, implied directly by PPT\textsuperscript{2}-conjecture. It remains an open question if it holds or not, though. 

\begin{theorem}
    Let $(\phi_t)_{t\in\reals_+}$ be a semigroup of linear maps on $\matrd$ and assume the PPT\textsuperscript{2}-conjecture holds. If there exists $s > 0$ such that $\phi_s$ lays inside the cone $\ppte{\matrd}$, then the semigroup is eventually EB.
\end{theorem}

\begin{proof}
    Again, from continuity we know there exists a non-empty interval $\mathcal{I} = [t_1, t_2] \subset (0,\infty)$ s.t.~$\phi_t \in \ppte{\matrd}$ for all $t\in\mathcal{I}$. Theorem \ref{thm:AsympEBsemigroup} then yields existence of such $t_* > 0$ that $(\phi_t)_{t\geqslant t_*}$ is PPT everywhere. PPT\textsuperscript{2}-conjecture then yields $\phi_t \in \ebe{\matrd}$ for all $t \geqslant 2t_*$, i.e.~the semigroup is eventually EB.
\end{proof}

This completes the more abstract part of the article. In what follows, we conduct analysis of some distinctive, important classes of quantum dynamical maps. We present our results in form of sections \ref{sec:CasesQubit} and \ref{sec:CasesBeyondQubit}, where the former is restricted solely to qubit cases ($d=2$) while the latter treats more general systems.

It is natural to ask \emph{when} eventually EB family actually becomes entanglement breaking. Let $X\subset B(\matrd)$ to be some subset of linear maps acting on $\matrd$. We define the \emph{X arrival time} $\atime{X}$ of family $(\phi_t)_{t\in\reals_+}$ as a minimum time after which the family enters subset $X$ and remains therein:
\begin{equation}
    \atime{X} = \min{\{t \, | \, \forall \, s \geqslant t : \phi_s \in X\}}.
\end{equation}
It is then justified to define a whole hierarchy of arrival times $\atime{CP}$, $\atime{coCP}$, $\atime{PPT}$, $\atime{EB}$ and so on. Naturally, by embeddings between cones of maps we have
\begin{equation}
    \atime{CP} \leqslant \atime{PPT} \leqslant \atime{EB} \quad \text{and} \quad \atime{coCP} \leqslant \atime{PPT} \leqslant \atime{EB}.
\end{equation}
In simple case of algebra $\matr{2}$ it suffices for a map to be PPT in order to be also EB: indeed, when Choi's matrix is PPT it follows from Peres-Horodecki criterion \cite{Peres1996,Horodecki1996} that it is also separable, i.e.~map is entanglement breaking, and $\atime{EB} = \atime{PPT}$ in this case. When $d > 2$ however, making exact calculation of $\atime{EB}$ is generally unmanageable and one must resort to finding the PPT arrival time $\atime{PPT}$ as its lower bound. This however can be achieved quite easily by analyzing definiteness of time-dependent partially transposed Choi's matrix of $\phi_t$, at least numerically.

\section{Examples: qubit case}
\label{sec:CasesQubit}

We first illustrate our analysis with two well known examples of qubit evolution: the \emph{Pauli channels} and \emph{phase covariant dynamics}. The latter one includes a basic semigroup case, a note on time-dependent generator and an interesting case of so-called \emph{eternally non-Markovian evolution}. For this section, we will use standard Pauli matrices $\sigma_i$ for orthogonal basis in $\matr{2}$:
\begin{equation}
    \sigma_1 = \left(\begin{array}{cc}0 & 1 \\ 1 & 0\end{array}\right), \quad \sigma_2 = \left(\begin{array}{cc}0 & -i \\ i & 0\end{array}\right), \quad \sigma_3 = \left(\begin{array}{cc}1 & 0 \\ 0 & -1\end{array}\right), \quad \sigma_4 = I.
\end{equation}

\subsection{Pauli channels}

The prominent Pauli channel is characterized in terms of its infinitesimal generator \cite{Chruscinski2013}
\begin{equation}\label{eq:PauliChannelGen}
    L_t (\rho) = \sum_{k=1}^{3} \gamma_k (t) (\sigma_k \rho \sigma_k - \rho),
\end{equation}
where coefficients $\gamma_k (t)$ are real. One easily checks that $L_t$ is diagonal in basis of Pauli matrices and so $\Lambda_t$ is
\begin{equation}\label{eq:PauliChannel}
    \Lambda_t = \sum_{k=1}^{4} \lambda_k (t) \proj{k}, \quad \proj{k}(\rho) = \frac{1}{2} (\tr{\sigma_k \rho}) \, \sigma_k ,
\end{equation}
for coefficients $\lambda_k (t)$ given as
\begin{align}\label{eq:PauliChannelEigenvalues}
    \lambda_1 (t) &= e^{-2 \left[\Gamma_{2}(t) + \Gamma_{3}(t)\right]}, \quad \lambda_2 (t) = e^{-2 \left[\Gamma_{1}(t) + \Gamma_{3}(t)\right]}, \\
    \lambda_3 (t) &= e^{-2 \left[\Gamma_{1}(t) + \Gamma_{2}(t)\right]}, \quad \lambda_4 (t) = 1 \nonumber
\end{align}
and
\begin{equation}\label{eq:PauliChannelGamma}
    \Gamma_k (t) = \int\limits_{0}^{t}\gamma_{k}(s)ds .
\end{equation}
It is shown in \cite{Chruscinski2013,Chruscinski2015} that a necessary and sufficient condition for a dynamical map \eqref{eq:PauliChannel} to be P-divisible is
\begin{equation}\label{eq:PauliChannelPdivCond}
    \gamma_1 (t) + \gamma_2(t) \geqslant 0, \quad \gamma_1 (t) + \gamma_3 (t) \geqslant 0, \quad \gamma_2 (t) + \gamma_3 (t) \geqslant 0
\end{equation}
for all $t\in\reals_+$, while $\gamma_i (t) \geqslant 0$ is necessary and sufficient for CP-divisibility.

\subsubsection{Semigroup}

In simplest case when all coefficients are constants, $\gamma_k (t) = \gamma_k$, we have $\Gamma_{k}(t) = \gamma_k t$, dynamical map \eqref{eq:PauliChannel} trivializes to a semigroup and \eqref{eq:PauliChannelPdivCond} is simply a condition for positivity. Asymptotic behavior of the channel is then easily seen to be determined by values of $\gamma_i + \gamma_j$:

\begin{theorem}\label{thm:PauliSemigroup}
    Denote $s_{ij} = \gamma_i + \gamma_j$. The following statements hold for positive Pauli semigroup $(e^{tL})_{t\in\reals_+}$:
    \begin{enumerate}
        \item \label{thm:PauliSemigroup1} If $s_{ij} > 0$ for all $i\neq j$ then the semigroup is eventually EB;
        \item \label{thm:PauliSemigroup2} If $s_{ij} = 0$ for just \emph{one} pair of indices $(i,j)$ then the semigroup is asymptotically EB;
        \item \label{thm:PauliSemigroup3} If $s_{ij} = s_{kl} = 0$ for \emph{two} different pairs of indices $(i,j)$ and $(k,l)$ then the semigroup is neither asymptotically EB, CP nor coCP.
    \end{enumerate}
\end{theorem}

\begin{proof}
    It is enough to check for properties of a map $\Lambda_\infty = \lim_{t\to\infty}\Lambda_t$. For statement \ref{thm:PauliSemigroup1}, notice that expressions \eqref{eq:PauliChannelEigenvalues} yield $\lambda_k (t) \to 0$ for all $k < 4$, so we have
    \begin{equation}
        \Lambda_\infty = \proj{4} = \proj{\omega},
    \end{equation}
    a projection onto maximally mixed state $\omega = \frac{1}{2}I$, which is also a stationary state i.e.~spans kernel of $L$. Semigroup is then eventually EB by Theorem \ref{thm:TheTheoremSemigroups}. For statement \ref{thm:PauliSemigroup2} notice that when exactly one $s_{ij} = 0$ then the asymptotic map $\Lambda_\infty$ will be a projection of rank $2$: indeed, with no loss of generality, assume $s_{12} = 0$, i.e.~$\gamma_1 = -\gamma_2$. Then, positivity conditions \eqref{eq:PauliChannelPdivCond} combined with $s_{13} > 0$, $s_{23} > 0$ imply $\gamma_3 > |\gamma_1|$ which result in
    \begin{equation}
        \Lambda_\infty = \proj{\omega} + \proj{3}.
    \end{equation}
    The Choi's matrix $\choi{\Lambda_\infty} = \operatorname{diag}{\{1,\,0,\,0,\,1\}}$ is separable, $\Lambda_\infty = E_{11} \otimes E_{11} + E_{22}\otimes E_{22}$, but not strictly positive definite: in fact, it lays on a boundary of cone $\matr{4}^{\mathrm{sep.}}$. With some effort, one can compute Choi's matrices $\choi{\Lambda_t}$ and $\choi{\Lambda_t}^{\ptrans}$ (which we omit here) and check that their minimal eigenvalues,
    \begin{align}
        \min{\spec{\choi{\Lambda_t}}} &= -e^{-2\gamma_3 t}\sinh{2|\gamma_1|t}, \\
        \min{\spec{\choi{\Lambda_t}^{\ptrans}}} &= -e^{-2\gamma_3 t} \cosh{2\gamma_1 t},\nonumber
    \end{align}
    are both negative for $t > 0$ and tend to $0$ as $t\to\infty$. This shows that in this case the semigroup is asymptotically PPT and therefore asymptotically EB. Finally, for the remaining statement \ref{thm:PauliSemigroup3}, assume with no loss of generality $s_{12} = s_{23} = 0$ so that $\gamma_1 = -\gamma_2 = \gamma_3$. Then,
    \begin{equation}
        \Lambda_\infty = \proj{\omega} + \proj{1} + \proj{3},
    \end{equation}
    and one easily checks that $\spec{\choi{\Lambda_\infty}} = \spec{\choi{\Lambda_\infty}^{\ptrans}} = \{-\frac{1}{2}, \frac{1}{2},\frac{3}{2}\}$,
    with $\frac{1}{2}$ of multiplicity $2$, so $\Lambda_\infty$ is neither CP nor coCP.
\end{proof}

\subsubsection{Note on time-dependent generator}

Qualitatively, the above analysis may be into some extent translated to the case of time-dependent generators under certain circumstances. We note that asymptotic behavior of $\Lambda_t$ heavily depends upon behavior of coefficients $\Gamma_k (t)$ as defined in \eqref{eq:PauliChannelGamma}, in principle on convergence of  $\int_0^t\gamma_k (t)dt$, and as such is a nontrivial task to trace in general. However, if $\Gamma_i (t) + \Gamma_j (t)$ are assumed to mimic behavior of functions $(\gamma_i + \gamma_j)t$, then one simply reproduces results from the semigroup case as the following theorem shows (we present it without proof as it is virtually the same as in the semigroup case):

\begin{theorem}\label{thm:PauliChannel}
    Denote $S_{ij}(t) = \Gamma_i (t) + \Gamma_j (t)$. The following statements hold for P-divisible family of Pauli channels $(\Lambda_t)_{t\in\reals_+}$:
    \begin{enumerate}
        \item \label{thm:PauliChannel1} If $S_{ij}(t) \to \infty$ for all $i\neq j$ then the family is eventually EB;
        \item \label{thm:PauliChannel2} If $S_{ij} (t) \to 0$ for just \emph{one} pair of indices $(i,j)$ and $S_{kl}(t) \to \infty$ for every other pair $(k,l)$ then the family is asymptotically EB;
        \item \label{thm:PauliChannel3} If $S_{ij}(t),S_{kl} \to 0$ for two different pairs $(i,j)$, $(k,l)$ and $S_{mn}(t)\to\infty$ for remaining pair $(m,n)$ then the family is neither asymptotically EB, CP nor coCP.
    \end{enumerate}
\end{theorem}

\begin{remark}
    The general case of time-dependent generator is naturally far more involved then semigroup case in the sense that expressions $\Gamma_k (t)$ may exhibit nontrivial asymptotic behavior as $t\to\infty$, depending on properties of underlying functions $\gamma_k (t)$. In principle then coefficients $\lambda_k (t)$ may tend possibly to any real number (or diverge at all) and so the asymptotic map $\Lambda_\infty$, if exists, may have a range of properties. We mark this as an interesting topic for further study.
\end{remark}

\subsubsection{Eternally non-Markovian channel}

As a special example of the above time-dependent generator, consider
\begin{equation}\label{eq:eternallyNonMarkovian}
    L_t(\rho) = \frac{\alpha}{2} \sum_{i=1}^{2} (\sigma_i \rho \sigma_i - \rho) - \frac{\alpha}{2}\tanh{t} (\sigma_3 \rho \sigma_3 - \rho),
\end{equation}
where $\alpha > 0$. Resulting dynamical map, in case $\alpha = 1$, was explored in \cite{Hall2014} as an example of \emph{eternally non-Markovian evolution}, being always CP yet never CP-divisible (cf. also \cite{Nina,Fabio-Sergey} and \cite{Breuer2016,Li2018} for more insight into non-Markovianity). Indeed, dynamics governed by \eqref{eq:eternallyNonMarkovian} is not CP-divisible, regardless of $\alpha$, because of negativity of $-\tanh{t}$. Curiously, it also provides an interesting example of an evolution family which is eventually EB, but not eEB-divisible:

\begin{theorem}
    Let $\alpha > 1$. Then, an \emph{eternally non-Markovian} family $(\Lambda_t)_{t\in\reals_+}$ governed by generator \eqref{eq:eternallyNonMarkovian} is eventually EB, but not eEB-divisible.
\end{theorem}

\begin{proof}
    With vectorization techniques, one finds a spectral decomposition of the generator, $L_t = \sum_{i=1}^{4} \mu_i (t) \proj{i}$, for
    \begin{equation}
        \mu_1 (t) = 0, \quad \mu_2(t) = \mu_3 (t) = \alpha (\tanh{t}-1), \quad \mu_3 (t) = -2\alpha ,
    \end{equation}
    and projections
    \begin{align}
        &\proj{1}(\rho) = \proj{\omega}(\rho) = (\tr{\rho})\,\omega, \quad \proj{2}(\rho) = (\tr{\sigma_+ \rho})\,\sigma_- ,\\ &\proj{3}(\rho) = (\tr{\sigma_- \rho})\,\sigma_+ , \quad \proj{4} = \frac{1}{2}(\tr{\sigma_3 \rho})\,\sigma_3 , \nonumber
    \end{align}
    where a stationary state of $L_t$ is $\omega = \frac{1}{2}I$, a maximally mixed state and $\sigma_\pm = \frac{1}{2}(\sigma_1 \pm i \sigma_2)$ as earlier. Integrating directly, one obtains a spectral decomposition $\Lambda_t = \sum_{i=1}^{4} \lambda_i (t) \proj{i}$ with spectrum
    \begin{equation}
        \lambda_1 (t) = 1, \quad \lambda_2 (t) = \lambda_3 (t) = e^{-\alpha t}\cosh^\alpha{t}, \quad \lambda_4 (t) = e^{-2\alpha t},
    \end{equation}
    yielding
    \begin{equation}
        \Lambda_\infty (\rho) = \lim_{t\to\infty}\Lambda_t (\rho) =\proj{\omega} + \frac{1}{2^\alpha}(\proj{2}+\proj{3}).
    \end{equation}
    It is not hard to compute
    \begin{equation}
        \choi{\Lambda_\infty} = \left(\begin{array}{cccc}2^{-1} & 0 & 0 & 2^{-\alpha} \\ 0 & 2^{-1} & 0 & 0 \\ 0 & 0 & 2^{-1} & 0 \\ 2^{-\alpha} & 0 & 0 & 2^{-1}\end{array}\right),
    \end{equation}
    as well as
    \begin{equation}
        \spec{\choi{\Lambda_\infty}} = \spec{\choi{\Lambda_\infty}^{\ptrans}} = \{2^{-1},\,2^{-1}\pm 2^{-\alpha}\},
    \end{equation}
    where $2^{-1}$ is of multiplicity $2$, which yields $\choi{\Lambda_\infty},\choi{\Lambda_\infty}^{\ptrans} > 0$ and therefore $\Lambda_\infty$ is an interior point of $\ebe{\matr{2}}$ by Peres-Horodecki criterion and by Lemma \ref{lemma:IntEB} (in Appendix \ref{app:MathematicalSupplement}), i.e.~dynamics is eventually EB. Again, by vectorization one can obtain the propagator $V_{t,s} = \Lambda_t \circ \Lambda_{s}^{-1}$ (which we omit here) and its Choi matrix
    \begin{equation}
        \choi{V_{t,s}} = \left(\begin{array}{cccc}\frac{1}{2}(1+e^{2\alpha (s-t)}) & 0 & 0 & e^{\alpha (s-t)}\frac{\cosh^\alpha{t}}{\cosh^\alpha{s}} \\ 0 & \frac{1}{2}(1-e^{2\alpha(s-t)}) & 0 & 0 \\ 0 & 0 & \frac{1}{2}(1-e^{2\alpha(s-t)}) & 0 \\ e^{\alpha (s-t)}\frac{\cosh^\alpha{t}}{\cosh^\alpha{s}} & 0 & 0 & \frac{1}{2}(1+e^{2\alpha (s-t)})\end{array}\right).
    \end{equation}
    Upon closer examination in turns out that when $t\to\infty$, the minimal eigenvalues of both $\choi{V_{t,s}}$ and $\choi{V_{t,s}}^{\ptrans}$ tend to the same expression $2^{-1} - 2^{-\alpha}e^{\alpha s} \cosh^{-\alpha}{s}$, which eventually becomes negative for all $\alpha \geqslant 1$. Therefore $V_{t,s}$ is not EB for large $t$ and dynamics fails to be eventually EB-divisible.
\end{proof}

Contrary to other examples, the asymptotic map $\Lambda_\infty$ is not a projection but rather a linear combination of projections because of specific time dependence of generator's spectrum.

\begin{remark}
    We note that in the original case $\alpha = 1$ we have $0\in\spec{\choi{\Lambda_\infty}}$ and map $\Lambda_\infty$ lays on the boundary of $\ebe{\matr{2}}$; this results in the evolution being only asymptotically EB (and asymptotically PPT as well), yet not eventually EB-divisible (nor eventually PPT-divisible).
\end{remark}

\subsection{Phase covariant dynamics}

Here we consider the phase covariant evolution in $\matr{2}$ which is one of most important and well-studied cases. Consider a following generator
\begin{equation}\label{eq:PhaseCovQubitL}
    L = -\frac{i\Omega}{2}\comm{\sigma_z}{\,\cdot\,} + \gamma_+ L_+ + \gamma_- L_- + \gamma_z L_z ,
\end{equation}
where $\Omega, \gamma_{\pm},\gamma_z \in \reals$ and
\begin{equation}
    L_{\pm}(\rho) = \sigma_\pm \rho \sigma_\mp - \frac{1}{2}\acomm{\sigma_\mp \sigma_\pm}{\rho}, \quad L_z (\rho) = \sigma_3 \rho\sigma_3 - \rho,
\end{equation}
with raising and lowering operators $\sigma_\pm$ defined via $\sigma_\pm = \frac{1}{2}(\sigma_1 \pm i \sigma_2)$. Now, $L$ generates CP semigroup if $\gamma_\pm,\gamma_z \geqslant 0$.  To generate a semigroup of positive maps \cite{Filippov2020} one requires $\gamma_\pm\geqslant 0$ together with

\begin{equation}\label{eq:PhCovQubitRates}
    \gamma_z + \frac{1}{2}\sqrt{\gamma_+ \gamma_-} \geqslant 0 .
\end{equation}
The semigroup $(\Lambda_t)_{t\in\reals_+}$ generated by such $L$ can be shown to read
\begin{equation}\label{eq:PhaseCovLambda}
    \Lambda_t (\rho) = \left(\begin{array}{cc} T_{11}(t) \rho_{11} + T_{12}(t)\rho_{22} & e^{-(\Gamma_\mathrm{T} + i \Omega)t}\rho_{12} \\ e^{-(\Gamma_\mathrm{T} - i\Omega)t}\rho_{21} & T_{21}(t)\rho_{11} + T_{22}(t)\rho_{22} \end{array}\right),
\end{equation}
where the time-dependent stochastic matrix $T_{ij}(t)$ is defined by

\begin{eqnarray}
    T(t) = \left( \begin{array}{cc} p_+ + p_- e^{-\Gamma_\mathrm{L} t} & p_+ (1- e^{-\Gamma_\mathrm{L} t}) \\ p_- (1 - e^{-\Gamma_\mathrm{L} t}) & p_- + p_+ e^{-\Gamma_\mathrm{L} t}
    \end{array} \right) ,
\end{eqnarray}
with
\begin{equation}\label{eq:PhConQubit_abc}
    p_+ = \frac{\gamma_+}{\gamma_+ + \gamma_-}, \quad p_- = \frac{\gamma_-}{\gamma_+ + \gamma_-},
\end{equation}
and longitudinal $\Gamma_\mathrm{L}$ and transversal $\Gamma_\mathrm{T}$ relaxation rates read
\begin{equation}\label{eq:PhConQubit_GammaLT}
    \Gamma_\mathrm{L} = \gamma_+ + \gamma_-, \quad \Gamma_\mathrm{T} = \frac{1}{2}(\gamma_+ + \gamma_-) + 2\gamma_z .
\end{equation}
It is evident that for any initial $\rho_0 \in \matr{2}$, $\tr{\rho_0}=1$, the matrix $\Lambda_t (\rho_0)$ asymptotically tends to stationary state $\omega$,
\begin{equation}\label{eq:PhCovQubitOmega}
    \omega = \lim_{t\to\infty}\Lambda_t (\rho_0) = \operatorname{diag}{\{p_+,\,p_-\}},
\end{equation}
i.e.~$\Lambda_t \to \proj{\omega}$.  The asymptotic properties of \eqref{eq:PhaseCovLambda} can be summarized in terms of a following
\begin{theorem}\label{thm:PhaseCovQubit}
    The following statements hold for semigroup $(\Lambda_t)_{t\in\reals_+}$ governed by generator \eqref{eq:PhaseCovQubitL}:
    \begin{enumerate}
        \item\label{thm:PhaseCovQubit:1} It is CP and eventually EB if $\gamma_\pm > 0$, $\gamma_z \geqslant 0$.
        \item\label{thm:PhaseCovQubit:2} It is CP and asymptotically EB if one of the rates $\gamma_+$, $\gamma_-$ is $0$ and $\gamma_z = 0$.
        \item\label{thm:PhaseCovQubit:3} It is positive, then CP and eventually EB if $\gamma_+,\gamma_- > 0$ and $-\frac{1}{2}\sqrt{\gamma_+\gamma_-} < \gamma_z < 0$.
        \item\label{thm:PhaseCovQubit:4} It is positive yet never CP if $\gamma_+ = \gamma_- > 0$ and $\gamma_z = -\frac{1}{2}\gamma_+$.
    \end{enumerate}
\end{theorem}

\begin{proof}
Ad \ref{thm:PhaseCovQubit:1}. Note that when $\gamma_\pm > 0$ the stationary state \eqref{eq:PhCovQubitOmega} is strictly positive definite and Theorem \ref{thm:LGKS} applies. For the remaining statements, let us first compute Choi matrix
\begin{equation}
    \choi{\Lambda_t} = \left(\begin{array}{cccc} p_+ + p_- e^{-\Gamma_\mathrm{L}t} & 0 & 0 & e^{-(\Gamma_\mathrm{T} + i\Omega)t} \\ 0 & p_- (1-e^{-\Gamma_\mathrm{L}t}) & 0 & 0 \\ 0 & 0 & p_+ (1-e^{-\Gamma_\mathrm{L}t}) & 0 \\ e^{-(\Gamma_\mathrm{T} - 2i\Omega)t} & 0 & 0 & p_- + p_+ e^{-\Gamma_\mathrm{L}t}
 \end{array}\right),
\end{equation}
where we again use notation \eqref{eq:PhConQubit_abc} and \eqref{eq:PhConQubit_GammaLT}.

Ad \ref{thm:PhaseCovQubit:2}. Without loss of generality put $\gamma_+ = 0$; then \eqref{eq:PhCovQubitRates} forces $\gamma_z = 0$ as well. Resulting generator is still in GKLS form, so $\Lambda_t$ is CP. Minimal eigenvalue of $\choi{\Lambda_t}^{\ptrans}$ may be found to be simply $- e^{-t\gamma_-}$ which remains negative for all $t$, i.e.~$\Lambda_t$ never becomes coCP nor entanglement breaking. This is not surprising: if one of $\gamma_+$, $\gamma_-$ is $0$ the projection $\proj{\omega}$ lays on the boundary of cone $\ebe{\matr{2}}$ by Lemma \ref{lemma:PomegaProperties} which is being approached but never reached by $\Lambda_t$.

Ad \ref{thm:PhaseCovQubit:3}. With some effort, one can compute the minimal eigenvalue
\begin{align}
    \lambda_{\mathrm{min}}(t) &= \min{\spec{\choi{\Lambda_t}}}\\
    &= \frac{1}{2}\left[ 1+e^{-\Gamma_\mathrm{L}t} - \sqrt{\frac{(\gamma_+ - \gamma_-)^2}{\Gamma_\mathrm{L}^{2}}( 1-e^{-\Gamma_\mathrm{L}t})^2 + 4 e^{-(\Gamma_\mathrm{L} + 4\gamma_z)t}}\right] \nonumber
\end{align}
and then notice
\begin{equation}
    \lambda_{\mathrm{min}}(0) = 0, \quad \left.\frac{d}{dt}\right|_{t=0}\lambda_{\mathrm{min}}(t) = 2\gamma_z
\end{equation}
so whenever $\gamma_z < 0$ the minimal eigenvalue is monotonically decreasing in some right neighborhood of $t=0$ and becomes negative in consequence. Therefore, $\Lambda_t$, while still positive, cannot be CP everywhere. However, condition $\gamma_\pm > 0$ again assures it becomes PPT and entanglement breaking (by Peres-Horodecki criterion) in finite time.

Ad \ref{thm:PhaseCovQubit:4}. Finally, in the extreme case when $\gamma_+ = \gamma_-$ and $\gamma_z = -\frac{1}{2}\sqrt{\gamma_+ \gamma_-} = -\frac{1}{2}\gamma_+$ we have $\lambda_{\mathrm{min}}(t) = \frac{1}{2}(e^{-2\gamma_+ t}-1)$ which is negative over $(0,\infty)$ so $\Lambda_t$ is never CP (except for $t=0$).
\end{proof}

\section{Examples: beyond qubit case}
\label{sec:CasesBeyondQubit}

\subsection{Pure decoherence}

Consider the following time-dependent qudit generator

\begin{equation}   \label{dec}
    L^{\mathrm{dec}}_t(\rho) = - i\comm{H(t)}{\rho} +  \sum_{i,j=1}^d a_{ij}(t) \left( E_{ii} \rho E_{jj} - \delta_{ij} \frac 12 \acomm{E_{ii}}{\rho} \right) ,
\end{equation}
where the $d \times d$ Hermitian matrix $a_{ij}(t)$ is positive definite, and  $H(t) = \sum_i h_i(t) E_{ii}$.
One finds

\begin{equation}
    L^{\mathrm{dec}}_{t}(E_{ij}) = \ell_{ij}(t) E_{ij} ,
\end{equation}
with $\ell_{ii}(t)=0$, and 
\begin{equation}
\ell_{ij}(t) = -i \left(h_i(t) - h_j(t)\right) + a_{ij}(t) - \frac12 \left( a_{ii}(t) + a_{jj}(t)\right)  \,\, \text{when} \,\,  i \neq j .
\end{equation}
The corresponding CP-divisible dynamical map reads
\begin{equation}
    \phi_t(E_{ij}) = \lambda_{ij}(t) E_{ij} \ , \quad \lambda_{ij}(t) := \exp{\int\limits_0^t \ell_{ij}(s) ds} , 
\end{equation}
and hence it can be represented via the Schur product $\phi_t(\rho) = D(t) \odot \rho$ with the time-dependent matrix $D(t)$,
\begin{equation}
    D_{ii}(t) = 1, \quad D_{ij}(t) = \lambda_{ij}(t) \,\, \text{when} \,\, i \neq j .
\end{equation}
Hence, the evolution is asymptotically EB only if $D(t) \to I$ when $t\to\infty$. Any nontrivial residual coherence $\lambda_{ij}(\infty)$ prevents dynamics to be asymptotically EB (the same applies for PPT property). It is, therefore, clear that $\phi_t$ is eventually  
EB only if $D(t)$ becomes fully diagonal at finite time. This, however, may happen only if the map $\phi_t$ is non-invertible (i.e. the corresponding generator is singular), cf. \cite{Erling,Sagnik,Ujan}. 

\begin{corollary} The map $(\phi_t)_{t \geq 0}$ is eEB-divisible if and only if there exists $t_* < \infty$ such that $D(t) = I$ for $t \geqslant t_*$.     
\end{corollary}

\subsection{Diagonally covariant dynamics}

A linear map $\phi$ is diagonally covariant if 
\begin{equation}
    \phi(UXU^\hadj) = U\phi(X)U^\hadj ,
\end{equation}
for all diagonal $d \times d$ unitary matrices $U$. Any diagonally covariant Markovian generator has the following form \cite{Chruscinski2022}
\begin{equation}
    L_t = L_{t}^{\mathrm{dec}} + L_{t}^{\mathrm{class}} ,
\end{equation}
where $L_{t}^{\mathrm{dec}}$ is defined in (\ref{dec}) and the \emph{classical} generator reads
\begin{equation}
    L_{t}^{\mathrm{class}}(\rho) = 
    \sum_{i\neq j}^d b_{ij}(t) \left( E_{ij} \rho E_{ji} - \frac 12 \acomm{E_{jj}}{\rho} \right) ,
\end{equation}
where the coefficients $b_{ij}(t)\geqslant 0$. It provides therefore generalization of pure decoherence dynamics. It is already clear from the analysis of the pure decoherence evolution that diagonally covariant dynamics is asymptotically EB if for any initial state $\rho$ the asymptotic state  $\phi_\infty(\rho)$ is diagonal, i.e. there is no asymptotic coherence. It is, therefore, clear that $\phi_t$ is eventually  
EB only if for all initial states the coherences of $\phi_t(\rho)$ are lost in finite time. 

\begin{corollary} The diagonally covariant map $(\phi_t)_{t \geq 0}$ is eEB-divisible if the dynamical map generated by the decoherence part of the generator $L^{\rm dec}_t$ 
is non-invertible and the asymptotic state of the evolution generated by the classical part $L_{t}^{\mathrm{class}}$ 
is of the full rank. 
\end{corollary}

\subsection{Generalized depolarizing channel}

Consider a generator
\begin{equation}\label{eq:GenDepL}
    L(\rho) = \gamma (\omega \tr{\rho} - \rho )
\end{equation}
where $\gamma > 0$ and $\omega\in\matrd^+$, $\tr{\omega}=1$. By direct check, $L$ nullifies the trace and $\omega$ is its eigenvector for eigenvalue $0$. The resulting semigroup is given by expression
\begin{equation}\label{eq:GenDepChannel}
    \Lambda_t (\rho) = e^{tL}(\rho) = e^{-\gamma t}\rho + (1-e^{-\gamma t})\omega \tr{\rho}.
\end{equation}

\begin{theorem}
    Semigroup \eqref{eq:GenDepChannel} is eventually EB when $\omega > 0$ and asymptotically EB when $0 \in \spec{\omega}$.
\end{theorem}

\begin{proof}
For convenience let us find the spectral decomposition of $L$ first. Assume $a\in\matrd$ is an eigenvector, $L(a) = \lambda a$, for $\lambda\neq 0$. Since $L$ nullifies the trace, $a$ necessarily lays in the subspace of all traceless matrices, or in the kernel of trace functional, $\ker{\mathrm{tr}} = \mathrm{tr}^{-1}(\{0\})$. Then, $\ker{\tr{}}$ is an eigenspace of $L$ corresponding to eigenvalue $-\gamma$ of multiplicity $\dim{\ker{\mathrm{tr}}} = d^2 - 1$. Hence, generator $L$ admits spectral decomposition
\begin{equation}
    L = 0\cdot\proj{\omega} -\gamma \proj{0}
\end{equation}
for $\proj{\omega} = (\tr{\cdot})\, \omega$ and $\proj{0}$ a projection onto $\ker{\mathrm{tr}}$. Note, that $L$ is not a \emph{normal} operator and projections $\proj{\omega}$ and $\proj{0}$ are not mutually orthogonal: $\omega$ is not proportional to identity, but is a linear combination of $I$ and traceless matrices (in fact, one quickly checks that, for example, $\omega = \frac{1}{d}(\tr{\omega})\, I + \sum_{i=1}^{d-1}\beta_i (E_{11}-E_{ii})$ for $\beta_i = \frac{1}{d}\tr{\omega}-\omega_{i+1}$ where $\omega_i$ are eigenvalues of $\omega$). This allows to re-express \eqref{eq:GenDepChannel},
\begin{equation}
    \Lambda_t = \proj{\omega} + e^{-\gamma t}\proj{0}
\end{equation}
so clearly $\Lambda_t \to \proj{\omega}$ as $t\to\infty$. Then, if $\omega > 0$ we know from Theorem \ref{thm:LGKSsemigroup} that $\proj{\omega}$ lays inside $\ebe{\matrd}$ and $\Lambda_t$ becomes EB.
\end{proof}

\begin{theorem}
    The lower bound $\atime{PPT}$ for EB arrival time of family \eqref{eq:GenDepL} is
    \begin{equation}\label{eq:GenDepLtau}
        \atime{PPT} = \frac{1}{\gamma}\ln{\left[1+\frac{1}{2}\left(\min_{i<j}{\omega_i \omega_j}\right)^{-\frac{1}{2}}\right]},
    \end{equation}
    where $\omega_i > 0$ are eigenvalues of $\omega$. Moreover, $\atime{PPT}$ attains the lowest possible value
\begin{equation}\label{eq:GenDepLtauMin}
    \atime{PPT,min.} =\min_{\omega}{\atime{PPT}} = \frac{1}{\gamma}\ln{\frac{d+2}{2}},
\end{equation}
when $\omega = \frac{1}{d} I$, i.e.~when evolution tends to a maximally mixed state.
\end{theorem}

\begin{proof}
Applying \eqref{eq:GenDepChannel} one quickly finds
\begin{equation}
    \choi{\Lambda_t}^{\ptrans} = e^{-\gamma t} \sum_{i,j=1}^{d^2}E_{ij}\otimes E_{ji} + (1-e^{-\gamma t})\cdot I \otimes \omega .
\end{equation}
After some work, characteristic polynomial of $\choi{\Lambda_t}^{\ptrans}$ can be checked to read
\begin{align}
    \det\big[\choi{&\Lambda_t}^{\ptrans} - \lambda I\big] = \\
     &e^{-\gamma t} \prod_{i=1}^{d}(1 + g_i - \lambda e^{\gamma t}) \prod_{i<j}\left[(g_i-\lambda e^{\gamma t})(g_j - \lambda e^{\gamma t})-1\right] \nonumber
\end{align}
for $g_i = (e^{\gamma t}-1)\omega_i$. From this we obtain general expressions for eigenvalues,
\begin{equation}\label{eq:lambdai}
    \lambda_i (t) = e^{-\gamma t} + (1-e^{-\gamma t})\omega_i 
\end{equation}
for $1\leqslant i \leqslant d$ and
\begin{equation}\label{eq:lambdaij}
    \lambda_{ij}(t) = \frac{1}{2}\left[(1-e^{-\gamma t})(\omega_i+\omega_j) \pm \sqrt{4e^{-2\gamma t}+(1-e^{-\gamma t})^2 (\omega_i-\omega_j)^2}\right]
\end{equation}
for $i<j$. One easily verifies that $\lim_{t\to\infty}\lambda_i(t),\,\lim_{t\to\infty}\lambda_{ij}(t) \in \spec{\omega}$ so if $\omega > 0$ then $\choi{\Lambda_t}^{\ptrans}$ becomes positive definite and semigroup becomes PPT and EB. On the other hand, if, say $\omega_1 = 0$, then eigenvalues $\lambda_{1j}(t)$ of form \eqref{eq:lambdaij} are negative and $\lambda_{1j}(t)\to 0$, i.e.~$\Lambda_t$ is only asymptotically PPT. Since the semigroup is completely positive, PPT arrival time $\tau$ is
\begin{equation}
    \tau = \max_{i<j}{t_{ij}},
\end{equation}
where $t_{ij}$ is a root of eigenvalues $\lambda_{ij}$ given by \eqref{eq:lambdaij} (we ignore eigenvalues of a form \eqref{eq:lambdai} since they are always positive). After some algebra, $\tau$ is found to be in the claimed form \eqref{eq:GenDepLtau}. Lemma \ref{lemma:ProbVector} (Appendix \ref{app:MathematicalSupplement}) then shows that its smallest possible value \eqref{eq:GenDepLtauMin} is attained when $\omega = \frac{1}{d} I$. 
\end{proof}

\subsection{Detailed balance}

An important and well-studied class of open quantum systems are those weakly interacting with a heat bath and satisfying the condition of quantum detailed balance (see \cite{Alicki2006a,Breuer2002} and references therein). They are characterized by generators of a form

\begin{align}\label{eq:LdetailedBalance}
    L(\rho) = -i\comm{H}{\rho} + \sum_{\alpha}\sum_{w \geqslant 0} \left[V_{\alpha w}\rho V_{\alpha w}^{\hadj} - \frac{1}{2}\acomm{V_{\alpha w}^{\hadj}V_{\alpha w}}{\rho}\right. \\
    +e^{-\beta w}\left. \left(V_{\alpha w}^{\hadj}\rho V_{\alpha w} - \frac{1}{2}\acomm{V_{\alpha w}V_{\alpha w}^{\hadj}}{\rho}\right)\right] \nonumber
\end{align}
where $\beta$ is an inverse temperature of the bath, $H=H^\hadj$ is an effective (physical) Hamiltonian of the system and $w$ are the \emph{Bohr frequencies} of $H$, i.e.~$w = \epsilon - \epsilon^\prime$ for some $\epsilon,\epsilon^\prime \in \spec{H}$. Operators $V_{\alpha w}$ are defined by relation $e^{iHt}V_{\alpha w}e^{-iHt} = e^{-i w t}V_{\alpha w}$; the presence of $e^{-\beta w}$ term is due to the KMS (Kubo-Martin-Schwinger) condition imposed on autocorrelation functions of the reservoir. One shows that whenever $0\in\spec{L}$ is of multiplicity $1$ the generator \eqref{eq:LdetailedBalance} satisfies so-called \emph{quantum detailed balance} condition with respect to a unique stationary Gibbs state
\begin{equation}
    \rho_\beta = \frac{e^{-\beta H}}{\tr{e^{-\beta H}}},
\end{equation}
which is also a stationary state of a semigroup $(e^{tL})_{t\in\reals_+}$. Moreover, we have $e^{tL}(\rho_0)\to \rho_\beta$ for any $\rho_0$ as $t\to\infty$, i.e.~a system returns to equilibrium determined by $\beta$. Then we have

\begin{theorem}
    Semigroup $(e^{tL})_{t\in\reals_+}$ generated by $L$ of form \eqref{eq:LdetailedBalance} is eventually EB and eEB-divisible.
\end{theorem}

\begin{proof}
    Clearly, $\spec{e^{-\beta H}} = \{e^{-\beta E} : E \in \spec{H}\}$ is positive and so the stationary state $\omega = \rho_\beta$ is strictly positive definite. Thus, Theorem \ref{thm:LGKSsemigroup} applies.
\end{proof}

\subsection{Periodic generators in Weak Coupling Limit}

Authors of \cite{Alicki2006,Szczygielski2013,Szczygielski2014} considered an open quantum system of dimension $d$, weakly coupled to external thermal reservoir and driven by some external energy source such that its self-Hamiltonian $H_t$ is periodic with period $T$. In such case, family $u_t$ of unitary maps generated by $H_t$, i.e.~satisfying Schr\"{o}dinger equation $\dot{u}_t = -i H_t u_t$ can be, by virtue of celebrated Floquet's theorem \cite{Chicone2006}, put in a product form
\begin{equation}
    u_t = p_t e^{-i\bar{H}t}
\end{equation}
where $p_t$ is unitary and periodic (with period $T$) and $\bar{H}$ is Hermitian. It was shown that the reduced density matrix $\rho_t$ of such system is governed, in Weak Coupling Limit regime and under some common approximations, by a time-local Markovian Master Equation $\dot{\rho}_t = L_t(\rho_t)$, where $L_t$ is time-periodic generator
\begin{equation}
    L_t = -i \comm{H_t}{\cdot} + P_t \circ K \circ P_{t}^{-1},
\end{equation}
with $K$ being a GKLS generator and $P_t (a) = p_t a  p_{t}^{\hadj}$; here, $H_t$ appearing in the commutator is to be understood as a properly ``renormalized'', physical Hamiltonian, including Lamb shift corrections due to influence from the reservoir. The dynamical map governed by such $L_t$ may be then shown to be of a product form
\begin{equation}\label{eq:LtPeriodiWCLLambda}
    \Lambda_t = P_t \circ e^{tX}, \quad X = -i \comm{\bar{H}}{\cdot} + K ,
\end{equation}
also inferred by Floquet's theorem due to periodicity of $L_t$. Here, both maps $P_t$ and $e^{tX}$ are CP and trace preserving, $X$ is of GKLS form and commutes with a derivation $-i \comm{\bar{H}}{\cdot}$.

\begin{theorem}\label{thm:LtPeriodicWCL}
    If $\ker{K} = \complexes\omega$ and $\omega > 0$ then a family \eqref{eq:LtPeriodiWCLLambda} is eventually EB and eEB-divisible.
\end{theorem}

\begin{proof}
    Since clearly $\omega > 0$ is a stationary state of a semigroup $(e^{tX})_{t\in\reals_+}$, we have $\Lambda_t \xrightarrow{a} \mathsf{Z}_t$, where $\mathsf{Z}_t = P_t \circ\proj{\omega}$ is periodic with period $T$ and $\proj{\omega}\in\operatorname{Int}{\ebe{\matrd}}$ by Lemma \ref{lemma:PomegaProperties}. One easily checks $\mathsf{Z}_t$ is a projection onto periodic state $\omega(t) = P_t (\omega)$,
    \begin{equation}
        \mathsf{Z}_t = \proj{\omega(t)} = (\tr{\cdot})\, \omega(t).
    \end{equation}
    Its Choi's matrix is $\choi{\proj{\omega(t)}} = I \otimes \omega(t)$. Since $\omega(t)$ and $\omega$ are related by similarity transformation, $\omega > 0$ iff $\omega(t) > 0$; therefore $\{\mathsf{Z}_t : t\in[0,T)\}$ also lays in a strict interior of $\ebe{\matrd}$. The family \eqref{eq:LtPeriodiWCLLambda} is then eventually EB by Theorem \ref{thm:eEBSufficientCondition}. This is equivalent to the fact that the ODE governed by such periodic $L_t$ admits a periodic limit cycle which is simply $\omega(t)$; all trajectories $\rho_t = \Lambda_t (\rho_0)$, $\rho_0 \in\matrd$, asymptotically tend to this cycle. Eventual EB-divisibility is then straightforward: by Theorem \ref{thm:SemigroupEB}, semigroup $(e^{tX})_{t\in\reals_+}$ is automatically eEB-divisible and the propagator $V_{t,s} = \Lambda_t \circ\Lambda_{s}^{-1}$, acting via
    \begin{equation}
        V_{t,s}(\rho) = P_t \circ e^{(t-s)X}\circ P_{s}^{-1}(\rho) = p_t e^{(t-s)X}(p_{s}^{\hadj}\rho p_s)p_{t}^{\hadj},
    \end{equation}
    also becomes EB as it differs from eEB-divisible semigroup only by additional compositions with completely positive maps.
\end{proof}

It is worth to remark that Theorem \ref{thm:LtPeriodicWCL} may be easily generalized to the case of \emph{quasiperiodic} Davies generators under additional assumption of Lyapunov-Perron reducibility of underlying Schr\"{o}dinger equation as introduced in \cite{Szczygielski2020}.

\section{Summary and open problems}
\label{sec:Summary}

We were able to show that a large class of quantum evolution families exhibits a tendency of becoming entanglement breaking or approaching the set of entanglement breaking maps. Those include some prominent cases of CP-divisible dynamics such as quantum dynamical semigroups, given positive definiteness of their respective stationary states. We also proposed a new notion of eventual divisibility, which may possibly find some applications in description of various systems which exhibit certain asymptotic behavior. Albeit we were able to prove a general asymptotic results in some simplified cases based on spectral properties of generators, we conjecture that the observation applies to much broader class of families. Possible further research directions include, but are not limited to, the following:

\begin{enumerate}
    \item Exploring asymptotic behavior of families governed by either non-commuting time-dependent generators in GKLS form or even by integro-differential Master Equations with non-trivial memory kernels (such as strongly non-Markovian ones).
    \item Investigating connections between various forms of eventual behavior of evolution families and various forms of eventual divisibility. Those include finding and exploring some interpolating examples of families which are, for instance, eventually PPT but not EB or ePPT-divisible but not eEB-divisible and so on.
    \item Characterizing certain forms of eventual divisibility and eventual behavior of evolution families by means of mathematical structure of generators.
    \item Finally, a deeper understanding of connection between asymptotics, in entanglement breaking terms or not, and PPT\textsuperscript{2}-conjecture could be of interest for both mathematical physics and quantum information theory.
\end{enumerate}

\section*{Data availability statement}

No new data were created or analysed in this study.

\section*{Acknowledgments}

Authors are thankful to Frederik vom Ende for his comments on initial version of the manuscript and to anonymous Referees for suggestions. D.~Chru\'{s}ci\'{n}ski was supported by the Polish National Science Center project No. 2018/30/A/ST2/00837. K.~Szczygielski worked at the Nicolaus Copernicus University in Toru\'{n} at time when this research was conducted and the manuscript written.
\appendix

\section{Mathematical supplement}
\label{app:MathematicalSupplement}

\begin{lemma}\label{lemma:CJIisometry}
    Let $X \subset B(\matrd)$ be closed in supremum norm topology. Then, $\phi\in \operatorname{Int}{X}$ ($\phi\in\partial X$, resp.) iff $\choi{\phi}\in\operatorname{Int}{\choi{X}}$ ($\choi{\phi}\in\partial \choi{X}$, resp.) in spectral matrix norm topology.
\end{lemma}

\begin{proof}
    Choi-Jamiołkowski isomorphism $\mathcal{C}$ induces a norm $\|\cdot\|_\mathcal{C}$ on $\matrd\otimes\matrd$ defined by $\| \cdot \|_{\mathcal{C}} = \supnorm{\cdot}\circ\mathcal{C}^{-1}$ which makes it a bijective isometry from $B(\matrd)$ to $(\matrd\otimes\matrd, \| \cdot \|_\mathcal{C})$. The topology induced by $\mathcal{C}$ on $\matrd\otimes\matrd$ is then precisely the norm topology by equivalence of norms; thus, claim follows.
\end{proof}

\begin{lemma}\label{lemma:IntEB}
    Let $\phi\in B(\matrd)$. Then $\phi\in\operatorname{Int}{\ebe{\matrd}}$ ($\phi\in\partial\ebe{\matrd}$, resp.) if and only if $\choi{\phi}\in\operatorname{Int}{\matr{d^2}^{\mathrm{sep.}}}$ ($\choi{\phi}\in\partial\matr{d^2}^{\mathrm{sep.}}$, resp.), where $\matr{d^2}^{\mathrm{sep.}}$ is a closed, convex set of separable matrices in $\matr{d^2}$.
\end{lemma}

\begin{proof}
    Choi-Jamiołkowski isomorphism is a bijection between sets of entanglement breaking maps and separable matrices, so the claim follows directly from Lemma \ref{lemma:CJIisometry}.
\end{proof}

\begin{lemma}\label{lemma:IntSep}
    We have $I\otimes\omega \in \operatorname{Int}{\matr{d^2}^{\mathrm{sep.}}}$ iff $\omega > 0$.
\end{lemma}

\begin{proof}
    First, let us assume $I\otimes\omega$ lays inside $\matr{d^2}^{\mathrm{sep.}}$, that is there exists an open ball $\mathcal{B}(I\otimes\omega, r)$ of some radius $r>0$, contained in $\matr{d^2}^{\mathrm{sep.}}$. Matrix $I\otimes\omega$, being a positive semi-definite, admits a factorization
    \begin{equation}
        I\otimes\omega = (I \otimes U)(I\otimes D)(I\otimes U^\hadj)
    \end{equation}
    for unitary $U$ and diagonal $D = \operatorname{diag}{\{\lambda_i\}_{i=1}^{d}}$ where all $\lambda_i \geqslant 0$. By way of contradiction, assume $\omega$ is not positive definite, i.e.~that it has at least one $0$ eigenvalue, say $\lambda_1 = 0$ (if $\omega$ has a negative eigenvalue or is non-Hermitian, $I\otimes\omega$ automatically is not positive semi-definite and not separable). Let us define $M = I\otimes UD_0U^\hadj$ where
    \begin{equation}
        D_0 = \operatorname{diag}{\{-\lambda, \, \lambda_2, \, ... \, , \, \lambda_{d}\}},
    \end{equation}
    where $\lambda\in (0,r)$ is arbitrary. Then, we see
    \begin{equation}
        \supnorm{I\otimes\omega - M} = \supnorm{D - D_0} = \lambda < r,
    \end{equation}
    so $M\in \mathcal{B}(I\otimes\omega, r)$. However, clearly $I\otimes M$ is not positive semi-definite, thus not separable. This means that there exists no open ball fully contained in $\matr{d^2}^{\mathrm{sep.}}$ when $0\in\spec{\omega}$, i.e.~$I\otimes\omega$ is not an interior point, a contradiction; therefore $\omega$ has to be positive definite. For the opposite, take $\omega > 0$, i.e.~all $\lambda_i > 0$ and notice that the above reasoning actually shows that when $0\in\spec{\omega}$, matrix $I\otimes\omega$ lays on the boundary of $\matr{d^2}^{\mathrm{sep.}}$. Indeed, let again $\lambda_1 = 0$ and define $N = UE_0 U^\hadj$ for
    \begin{equation}
        E_0 = \operatorname{diag}{\{\lambda, \, \lambda_2, \, ... \, , \, \lambda_{d}\}},
    \end{equation}
    and $M = UD_0 U^\hadj$ as earlier. Then, for $\lambda\in(0,r)$, we have $\supnorm{I\otimes\omega - M} < r$ and $\supnorm{I\otimes\omega - N} < r$ as can be checked, so both $M,N$ lay inside open ball $\mathcal{B}(I\otimes\omega,r)$ for all $r > 0$. By construction, $I\otimes N$ is separable while $I\otimes M$ is not; this shows $I\otimes\omega$ lays on the boundary whenever $\omega$ has a zero eigenvalue. Therefore, when $\omega > 0$, matrix $I\otimes\omega$ either lays inside the set, or completely outside. The latter would however imply $I\otimes\omega$ is not separable, which is absurd; hence, $I\otimes\omega$ must lay inside $\matr{d^2}^{\mathrm{sep.}}$ and the proof is complete.
\end{proof}

\begin{lemma}\label{lemma:PomegaProperties}
    Let $\omega\in\matrd$ satisfy $\tr{\omega} = 1$. Define a map $\proj{\omega}:\matrd\to\complexes\omega$ by
    \begin{equation}
        \proj{\omega}(a) = (\tr{a})\,\omega .
    \end{equation}
    Then, the following statements hold:
    \begin{enumerate}
        \item\label{lemma:PomegaProperties:PointRay} $\proj{\omega}$ is a rank-one projection,
        \item\label{lemma:PomegaProperties:PointPPT} $\proj{\omega}\in\ebe{\matrd}$ iff $\omega\in\matrd^+$,
        \item\label{lemma:PomegaProperties:PointInterior} $\proj{\omega}\in\operatorname{Int}{\ebe{\matrd}}$ iff $\omega > 0$,
        \item\label{lemma:PomegaProperties:PointBoundary} otherwise, if $\omega\in\matrd^+$ but is not strictly positive definite, $\proj{\omega}$ lays in the intersection of $\partial \ebe{\matrd}$, $\partial\cpe{\matrd}$ and $\partial\cocpe{\matrd}$.
    \end{enumerate}
\end{lemma}

\begin{proof}
    Readily $\operatorname{Im}{\proj{\omega}} = \{z \omega : z\in\complexes\}$ is of dimension 1 and checking $\proj{\omega}$ is idempotent is straightforward; hence statement \ref{lemma:PomegaProperties:PointRay} follows. Statement \ref{lemma:PomegaProperties:PointPPT} is immediate since Choi's matrix $\choi{\proj{\omega}} = I\otimes\omega$ is separable iff $\omega \geqslant 0$. Then, it lays inside set of separable matrices iff $\omega > 0$ by Lemma \ref{lemma:IntSep} and Lemma \ref{lemma:IntEB} yields statement \ref{lemma:PomegaProperties:PointInterior}. Finally, statement \ref{lemma:PomegaProperties:PointBoundary} is a direct consequence of all previous ones. Readily, when $0\in\spec{\omega}$, Lemma \ref{lemma:IntEB} yields $\proj{\omega}$ lays on the boundary of $\ebe{\matrd}$. Choi matrices $\choi{\proj{\omega}} = I\otimes\omega$ and $\choi{\proj{\omega}}^{\ptrans} = I\otimes\omega^\transpose$ are mutually positive semi-definite iff $\omega \geqslant 0$, so $\proj{\omega}$ is CP iff it is coCP. Then it is easy to see that when $0\in\spec{\omega}$, matrices $I\otimes\omega$ and $I\otimes\omega^\transpose$ lay on the boundary of $(\matrd\otimes\matrd)^+$ so $\proj{\omega}$ lays on boundaries of both $\cpe{\matrd}$ and $\cocpe{\matrd}$.
\end{proof}

\begin{lemma}\label{lemma:Halfline}
    For every non-empty interval $[a,b] \subset (0,\infty)$ there exists a half-line $[x_0,\infty)$, $x_0>0$, s.t.~a family $\{[na, nb] : n>n_0\}$ is its covering for $n_0 \in\naturals$ large enough.
\end{lemma}
\begin{proof}
Denote $\mathcal{I}_n = [na,nb]$. Le us assume that there is some $x \in (0,\infty)$ which does not belong to the union $\mathcal{U} = \bigcup_{n\in\naturals}\mathcal{I}_n$, that is $x\in (0,\infty)\setminus\mathcal{U}$. This means there exists some index $k$ s.t.~$x$ lays between intervals $\mathcal{I}_k$ and $\mathcal{I}_{k+1}$, i.e.~in open interval $(kb,(k+1)a)$. This interval however itself must be non-empty, that is $kb < (k+1)a$ which is possible iff $k \in [1, \frac{a}{b-a}]\cap\naturals$ as can be easily checked. This means that there exists only a finite set of possible indices $k$ which guarantee existence of such $x$; one checks $k \leqslant \big\lceil \frac{a}{b-a} \big\rceil -1$. This however means that $x < (k+1)a \leqslant \big\lceil \frac{a}{b-a} \big\rceil a$ and so possible values of such elements $x$ are upper bounded by $x_0 = \big\lceil \frac{a}{b-a} \big\rceil a$. By contraposition, for $x \geqslant x_0$ we have $x \in \mathcal{U}$, i.e.~$\mathcal{U}$ covers the half-line $[x_0,\infty )$.
\end{proof}

\begin{lemma}\label{lemma:ProbVector}
    Let $\vec{p}\in\reals_{+}^{n}$ be a probability vector, i.e.~$p_i \in [0,1]$, $\|\vec{p}\|_1 = \sum_{i=1}^{n}p_i = 1$. Then, a function
    \begin{equation}
        f(\vec{p}) = \min_{i<j}{p_i p_j}
    \end{equation}
    attains its maximum value $n^{-2}$ at uniform probability distribution.
\end{lemma}

\begin{proof}
    Let $\vec{p}_0 = \left(\frac{1}{n}, ... ,\frac{1}{n}\right)$ denote the uniform probability vector. Clearly, $f(\vec{p}_0) = n^{-2}$ and we will show inductively that $f(\vec{p})\leqslant f(\vec{p}_0)$. First, check that for $n=2$,
    \begin{equation}
        f(\vec{p}) = p_1 p_2 = p_1(1-p_1) \leqslant \frac{1}{4},
    \end{equation}
    where equality holds only for $p_1 = \frac{1}{2}$ so the base case is trivially true. Second, consider $\vec{p}\in\reals_{+}^{n+1}$ and arrange its components in non-decreasing order so that $p_1 \leqslant p_2 \leqslant ... \leqslant p_{n+1}$ and therefore
    \begin{equation}
        f(\vec{p}) = p_1 p_2 .
    \end{equation}
    Now, if $\vec{p}$ is strictly distinct from uniform distribution, the normalization condition $\|\vec{p}\|_1 = 1$ yields that necessarily $p_1 < \frac{1}{n+1}$ and $p_{n+1} > \frac{1}{n+1}$ (with possibly more components $p_i$ being different from $\frac{1}{n+1}$). We can always write $\vec{p}$ as $\vec{p} = (\vec{r},p_{n+1})$ where $\vec{r} = (p_i) \in \reals_{+}^{n}$ and $\|\vec{r}\|_1 = 1 - p_{n+1}$. For such $\vec{r}$ define
    \begin{equation}
        \tilde{\vec{r}} = \frac{\vec{r}}{1-p_{n+1}},
    \end{equation}
    so that $\|\tilde{\vec{r}}\|_1 = 1$. By induction hypothesis,
    \begin{equation}
        f(\tilde{\vec{r}}) = \frac{p_1 p_2}{(1-p_{n+1})^2} \leqslant \frac{1}{n^2}
    \end{equation}
    (note that it may still happen that $\tilde{\vec{r}}$ is a uniform distribution); this yields
    \begin{equation}
        f(\vec{p}) = p_1 p_2 = (1-p_{n+1})^2 f(\tilde{\vec{r}}) \leqslant \left( \frac{1-p_{n+1}}{n}\right)^2 .
    \end{equation}
    However, as $p_{n+1} > \frac{1}{n+1}$, simple algebra shows that
    \begin{equation}
        \frac{1-p_{n+1}}{n} < \frac{1}{n+1}
    \end{equation}
    and so $p_1 p_2 < (n+1)^{-2}$ whenever $\vec{p} \neq \left( \frac{1}{n+1}, ... ,\frac{1}{n+1} \right)$ and the proof is finished.
\end{proof}

\bibliographystyle{unsrt}
\bibliography{EB_Bibliography}

\end{document}